\documentclass[lettersize,journal]{IEEEtran}

\usepackage{array}
\usepackage[caption=false,font=normalsize,labelfont=sf,textfont=sf]{subfig}
\usepackage{textcomp}
\usepackage{url}
\usepackage{verbatim}

\usepackage{amsmath,amssymb,amsfonts,mathtools,bm,color}
\usepackage{cite}
\usepackage{graphicx}
\usepackage{algorithm}
\usepackage{algpseudocode}
\usepackage{cases}

\newtheorem{theorem}{Theorem}
\newtheorem{proposition}{Proposition}

\newtheorem{lemma}{Lemma}
\newtheorem{remark}{Remark}
\newtheorem{assumption}{Assumption}
\newtheorem{Fact}{Fact}

\newcommand{\Hcal}{\mathcal H}
\newcommand{\Ycal}{\mathcal Y}
\newcommand{\Scal}{\mathcal S}
\newcommand{\Qcal}{\mathcal Q}
\newcommand{\Acal}{\mathcal A}
\newcommand{\Rcal}{\mathcal R}

\newcommand{\Zcal}{\mathcal Z}
\newcommand{\Bcal}{\mathcal B}
\newcommand{\Pcal}{\mathcal P}
\newcommand{\Ical}{\mathcal I}
\newcommand{\Ucal}{\mathcal U}
\newcommand{\Tr}{\operatorname{Tr}}
\newcommand{\rank}{\operatorname{rank}}
\newcommand{\ran}{\operatorname{range}}
\newcommand{\spann}{\operatorname{span}}
\newcommand{\argmax}{\operatorname*{arg\,max}}
\newcommand{\norm}[1]{\left\lVert #1\right\rVert}
\newcommand{\inner}[2]{\left\langle #1,#2\right\rangle}
\def\BibTeX{{\rm B\kern-.05em{\sc i\kern-.025em b}\kern-.08em
    T\kern-.1667em\lower.7ex\hbox{E}\kern-.125emX}}
\usepackage{balance}
\begin{document}
\title{CAPA Beamforming for Extended-Target ISAC: Finite-Rank Solution Structure and Adaptive Electromagnetic-Mode Enrichment}
\author{Hongli Liu and Qiang Li
	\thanks{H. Liu and Q. Li are with School of Information and Communication  Engineering, University of Electronic Science and Technology of China, Chengdu, P.~R.~China, 611731. ({\it Corresponding author: Q. Li, E-mail: lq@uestc.edu.cn.})
	}
}


\maketitle

\begin{abstract}
Continuous-aperture arrays (CAPAs) provide a new electromagnetic transmission paradigm by directly controlling the source-current distribution over a continuous radiating aperture. This paper studies CAPA-enabled integrated sensing and communication (ISAC) with spatially extended sensing targets, for which the electromagnetic response over each target region is characterized by an aperture-to-region operator. Unlike point-target models, the resulting illumination operator may have infinite rank, precluding a priori finite-dimensional reduction of the transmit design. We formulate the joint communication and target-illumination beamforming problem in the continuous electromagnetic domain and develop an operator semidefinite relaxation (SDR). A key structural result is that the optimal transmit covariance operators remain finite rank even when the aggregate illumination operator has infinite rank. We further establish the tightness of the operator SDR and provide a constructive rank-one recovery of the information beamformers. Since the finite electromagnetic subspace supporting the optimum is generally unknown, we develop a dual-slack-driven adaptive mode-enrichment method that discovers the required modes from the full-space optimality conditions. The dual-slack violation simultaneously identifies missing modes and yields a computable certificate of the remaining full-space optimality gap. Global convergence of the enrichment sequence to the full-space optimum is established. Numerical results validate the finite-rank optimal structure and show that the proposed method approaches the full-space optimum with substantially fewer electromagnetic modes than truncated-SVD and spatial-Fourier representations.
\end{abstract}

\begin{IEEEkeywords}
CAPA, ISAC, extended target, infinite-dimensional optimization, semidefinite relaxation, electromagnetic-mode enrichment.
\end{IEEEkeywords}

\vspace{-5pt}
\section{Introduction}
\label{sec:introduction}
Multiple-antenna systems have evolved toward increasingly rich spatial degrees of freedom. This trend has driven wireless transceivers from conventional MIMO to massive and extremely large-scale arrays, and more recently toward densely integrated electromagnetic apertures~\cite{wang2024tutorial}. Continuous-aperture arrays (CAPAs) represent a natural limit of this evolution by replacing discrete antenna-port excitation with spatially continuous control of the source-current distribution over an electrically large aperture~\cite{ouyang2025diversity,liu2025capa}. Such continuous current control can be approached physically through densely coupled radiating structures, including programmable metasurfaces and related electrically, optically, or acoustically driven architectures~\cite{rezvani2025energy,zhao2025continuous,liu2025capa}. By controlling the current distribution over the radiating surface, CAPAs provide rich degrees of freedom for electromagnetic-field synthesis and offer a promising architecture for communication and sensing systems.

This transition from discrete arrays to continuous apertures fundamentally changes the mathematical nature of transmit design. Conventional beamforming optimizes finite-dimensional vectors or matrices, whereas CAPA beamforming is naturally posed as a functional optimization problem over an infinite-dimensional space, for which the matrix-based optimization tools developed for discrete arrays are no longer directly applicable. A common remedy is to first obtain a finite-dimensional representation of the continuous-domain problem, either by truncating or approximating the electromagnetic response, or by identifying an exact finite-dimensional subspace that preserves optimality. Existing CAPA communication studies have explored both directions. Prescribed Fourier-, wavenumber-, and pattern-domain representations have been used to obtain tractable finite-dimensional descriptions of spatially continuous apertures~\cite{pizzo2022fourier,sanguinetti2023wavenumber,zhang2023pattern}, while continuous beamformers have also been optimized directly in function space~\cite{wang2025beamforming,wang2025continuous,ranasinghe2026doubly}. Of particular relevance are recent subspace results for multiuser CAPA beamforming, which show that the optimal beamformers can be confined to the span of a finite number of user electromagnetic responses~\cite{wang2025optimal,guo2026deep,zhang2026transformer}. These results show that, although the aperture itself is continuous and the ambient design space is infinite dimensional, an exact finite-dimensional reduction is possible when the optimization depends only on a finite-rank electromagnetic response.

Recent CAPA-enabled integrated sensing and communication (ISAC) studies largely inherit this favorable structure. Existing designs have addressed communication--sensing tradeoffs, continuous-domain signal processing, multi-target illumination, and symbol-level waveform optimization~\cite{zhang2025exploiting,zhang2025RateCRB,zhao2026downlink,zhao2026continuous,ye2026optimal,liu2026symbol}. In most of these formulations, sensing is specified through a finite collection of target locations, directions, or spatial responses. Together with the finite number of communication users, these responses induce a finite-dimensional electromagnetic subspace within which the continuous-aperture optimization can be carried out.

Spatially extended targets lead to a different setting. At sufficiently high sensing resolution, different parts of a target give rise to resolvable scattering responses, so the conventional point-target abstraction becomes inadequate and the target extent needs to be explicitly modeled~\cite{Garcia2022crb}. Extended-target ISAC has consequently received increasing attention~\cite{du2023integrated,wang2024cramer,wang20263d,zhao2026crbET}, but existing studies mainly rely on discrete arrays or finite-dimensional target parameterizations. For a CAPA, retaining the electromagnetic field over an entire target region leads to an aperture-to-region operator. The associated illumination operator can have infinite rank, so the finite response subspace available in point-target formulations no longer exists a priori. This removes the structural shortcut that underpins many existing CAPA beamforming methods and leaves open a basic question: whether an infinite-rank sensing response necessarily requires an infinite-dimensional transmit solution.

This paper answers this question for CAPA-enabled extended-target ISAC. We formulate the joint communication and target-illumination design directly in the continuous electromagnetic domain and show that, despite the infinite-rank response, the optimal transmit covariance operators retain a finite-rank structure. As such, the difficulty of the transmit optimization shifts from representing infinitely many electromagnetic modes to identifying the unknown finite-dimensional modal support of the optimum. 
To this end, we develop an adaptive electromagnetic-mode enrichment framework based on an operator semidefinite relaxation (SDR), where the restricted problem is solved over the current modal space and the resulting full-space optimality information is used to select additional electromagnetic directions. The resulting sequence enriches the search space with optimization-relevant electromagnetic directions and converges to the full-space optimum.

The main contributions are summarized as follows.

\begin{itemize}
	
	\item We introduce a continuous-domain CAPA-ISAC formulation for spatially extended targets, where the electromagnetic field over each target region is modeled through an aperture-to-region operator. This moves beyond the finite collection of point responses considered in many existing CAPA-ISAC formulations and gives rise to an infinite-rank illumination regime for which no exact finite-dimensional response space is available a priori.
	
	\item We reveal that an infinite-rank illumination response does not imply an infinite-dimensional optimal transmission. The optimal transmit covariance operators remain finite rank, and the operator SDR is shown to be tight with a constructive recovery of globally optimal rank-one information beamformers.
	
	\item We develop a dual-slack-driven electromagnetic-mode enrichment method that recovers the unknown optimal modal support without prescribing a basis or truncation order. A novel full-space optimality certificate is established to guide mode discovery and quantify the remaining optimality gap. We further establish global convergence of the resulting enrichment sequence.
\end{itemize}

Numerical results illustrate the finite-rank structure of the optimized transmit covariance and show that the proposed method approaches the full-space optimum with substantially fewer modes than truncated-SVD and spatial-Fourier representations.

\begin{table}[htp!]
	\centering
	\caption{Operator Notation}
	\label{tab:notation}
	\renewcommand{\arraystretch}{1.12}
	\begin{tabular}{p{0.29\columnwidth} p{0.63\columnwidth}}
		\hline
		Symbol & Meaning \\
		\hline
		
		$\Acal^*$
		& Adjoint of operator $\Acal$ \\
		
		$\Tr(\Acal)$
		& Trace of a trace-class operator $\Acal$ \\
		
		$\rank(\Acal)$
		& Dimension of the range of $\Acal$ \\
		
		$\mathcal S_1(\mathcal X)$
		& Space of trace-class operators on $\mathcal X$ \\
		
		$\mathcal S_1^+(\mathcal X)$
		& Cone of PSD trace-class operators on $\mathcal X$ \\
		
		$\mathcal K(\mathcal X)$
		& Space of compact operators on $\mathcal X$ \\
		
		$\Acal\succeq\Bcal$
		& $\Acal-\Bcal$ is positive semidefinite \\
		
		$\Ical$
		& Identity operator on the ambient Hilbert space \\
		
		$\ran(\Acal)$, $\ker(\Acal)$
		& Range and nullspace of $\Acal$ \\
		
		$\mathcal M^\perp$
		& Orthogonal complement of subspace $\mathcal M$ \\
		
		$\mathcal M_1\oplus\mathcal M_2$
		& Orthogonal direct sum \\
		\hline
	\end{tabular}
	\vspace{-5mm}
\end{table}
\vspace{-5pt}
\subsection{Notation and Mathematical Preliminaries}
\label{subsec:notation}

Scalars are denoted by italic lowercase letters $a$, vectors and vector-valued functions by bold lowercase letters $\bm a$, finite-dimensional matrices by bold uppercase letters $\bm A$, and Hilbert spaces and infinite-dimensional operators by calligraphic letters $\Hcal$ and $\Acal$, respectively. The sets of real and complex numbers are denoted by $\mathbb R$ and $\mathbb C$. The superscripts $(\cdot)^T$ and $(\cdot)^H$ denote the transpose and conjugate transpose, respectively. We use $\mathbb E[\cdot]$ for expectation, $\norm{\cdot}_2$ for the Euclidean norm, $\norm{\cdot}_{\rm op}$ for the operator norm, and $\norm{\cdot}_{\rm HS}$ for the Hilbert--Schmidt norm. The  operator-theoretic notation follows standard functional-analysis conventions~\cite{kreyszig1978introductory}.

For a spatial domain $\mathcal D\subseteq\mathbb R^p$, the vector-valued space
\begin{equation*}
	\mathcal X
	=
	L^2(\mathcal D;\mathbb C^d)
	=
	\left\{
	\bm a:\mathcal D\rightarrow\mathbb C^d:
	\int_{\mathcal D}\norm{\bm a(\bm z)}_2^2\,d\bm z<\infty
	\right\},
\end{equation*}
is a Hilbert space equipped with the inner product
\begin{equation*}
	\inner{\bm a}{\bm b}_{\mathcal X}
	=
	\int_{\mathcal D}
	\bm a^H(\bm z)\bm b(\bm z)\,d\bm z,
	\qquad
	\norm{\bm a}_{\mathcal X}^2
	=
	\inner{\bm a}{\bm a}_{\mathcal X},
\end{equation*}
where the inner product is conjugate linear in its first argument. For a bounded linear operator $\Acal:\mathcal X\rightarrow\mathcal Y$, its adjoint $\Acal^*:\mathcal Y\rightarrow\mathcal X$ satisfies
\begin{equation*}
	\inner{\Acal\bm x}{\bm y}_{\mathcal Y}
	=
	\inner{\bm x}{\Acal^*\bm y}_{\mathcal X}.
\end{equation*}
An operator $\Acal:\mathcal X\rightarrow\mathcal X$ is self-adjoint if $\Acal=\Acal^*$ and positive semidefinite (PSD), denoted $\Acal\succeq0$, if $\inner{\bm x}{\Acal\bm x}_{\mathcal X}\ge0$ for all $\bm x\in\mathcal X$. Accordingly, $\Acal\succeq\Bcal$ means $\Acal-\Bcal\succeq0$.

The rank of an operator is $\rank(\Acal)=\dim\ran(\Acal)$, which may be finite or infinite. For $\bm a,\bm b\in\mathcal X$, the rank-one operator $\bm a\otimes\bm b$ is defined by
\begin{equation}
	(\bm a\otimes\bm b)\bm u
	=
	\bm a\inner{\bm b}{\bm u}_{\mathcal X},
	\qquad
	\bm u\in\mathcal X.
	\label{eq:rank_one_operator}
\end{equation}

An operator $\Acal$ is trace class if its singular values $\{s_n(\Acal)\}$ are summable, with
\(
\norm{\Acal}_1\triangleq\sum_{n=1}^{\infty}s_n(\Acal)<\infty
\).
The space of trace-class operators on $\mathcal X$ is denoted by $\mathcal S_1(\mathcal X)$, with $\mathcal S_1^+(\mathcal X)$ denoting its PSD cone.
For $\Acal\in\mathcal S_1(\mathcal X)$,
\begin{equation}\label{eq:tr_def_1}
\Tr(\Acal)
=
\sum_{n=1}^{\infty}
\inner{\bm e_n}{\Acal\bm e_n}_{\mathcal X},
\end{equation}
where $\{\bm e_n\}$ is any orthonormal basis of $\mathcal X$.
If $\Acal\succeq0$,
\(
\Tr(\Acal)=\norm{\Acal}_1
=\sum_{n=1}^{\infty}\lambda_n(\Acal)
\).

\subsection{Organization}
\label{subsec:organization}

The remainder of this paper is organized as follows. Section~\ref{sec:system_model} introduces the CAPA communication and extended-target electromagnetic models and formulates the continuous-domain beamforming problem. Section~\ref{sec:method} develops the operator SDR and establishes its structural properties and tightness. Section~\ref{sec:mode_enrichment} presents the adaptive electromagnetic-mode enrichment method and its convergence analysis. The numerical results are presented in Section~\ref{sec:simulation}, followed by the conclusion in Section~\ref{sec:conclusion}.

\section{System Model and Formulation}
\label{sec:system_model}


Consider the downlink CAPA-enabled ISAC system in Fig.~\ref{fig:system_model}, where a base station (BS) equipped with a planar continuous transmit aperture serves \(K\) communication users while illuminating \(Q\) spatially extended sensing targets. The transmit current carries \(K\) independent information streams and a dedicated sensing signal. We next present the communication and extended-target sensing models and formulate the corresponding beamforming problem.

\begin{figure}[t]
	\centering
	\includegraphics[width=0.7\linewidth,trim={0.0cm 0.3cm 0.0cm 0.2cm}, clip]{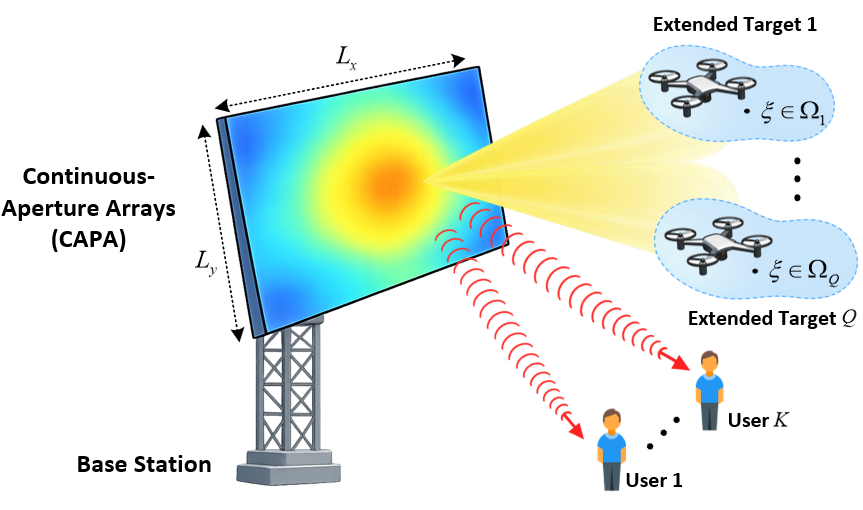}
	\caption{CAPA-enabled downlink ISAC system.}
	\label{fig:system_model}
	\vspace{-2mm}
\end{figure}
\subsection{CAPA Transmission and Communication Model}
\label{subsec:tx_model}

The CAPA transmitter is modeled as a continuous radiating aperture \(\Scal_T\) lying in the \(x\)--\(y\) plane, with side lengths \(L_x\) and \(L_y\). The position of an arbitrary point on the aperture is denoted by \(\bm s=[s_x,s_y,0]^T\in\mathbb R^3\), and the aperture domain is
\begin{equation}
	\Scal_T
	=
	\left\{
	[s_x,s_y,0]^T:
	|s_x|\le\frac{L_x}{2},\;
	|s_y|\le\frac{L_y}{2}
	\right\}
	\subset\mathbb R^3.
	\label{eq:aperture_domain}
\end{equation}
The equivalent source current at \(\bm s\in\Scal_T\) is represented by the three-component vector
\begin{equation}
	\bm j(\bm s)
	=
	[j_x(\bm s),j_y(\bm s),j_z(\bm s)]^T
	\in\mathbb C^3,
	\label{eq:vector_current}
\end{equation}
where \(j_x(\bm s)\), \(j_y(\bm s)\), and \(j_z(\bm s)\) denote the current components along the three Cartesian directions. Accordingly, the admissible aperture currents belong to the Hilbert space
\begin{equation}
	\Hcal
	=
	L^2(\Scal_T;\mathbb C^3).
	\label{eq:current_space}
\end{equation}

The transmit current consists of \(K\) independent information streams and a dedicated sensing signal, and is expressed as 
\begin{equation}
	\bm j(\bm s)
	=
	\sum_{k=1}^{K}
	\bm w_k(\bm s)x_k
	+
	\bm v(\bm s),
	\qquad
	\forall~\bm s\in\Scal_T,
	\label{eq:tx_current}
\end{equation}
where \(\bm w_k\in\Hcal\) denotes the information beamformer for user \(k\), and \(x_k\) is the corresponding information symbol. 
The information symbols are assumed to satisfy \(\mathbb E[x_k]=0\) and \(\mathbb E[x_kx_i^*]=\delta_{ki}\). The dedicated sensing signal \(\bm v\) is modeled as a zero-mean, second-order \(\Hcal\)-valued random element independent of \(\{x_k\}_{k=1}^{K}\), with covariance operator
\begin{equation}
	\Qcal_s
	=
	\mathbb E[\bm v\otimes\bm v]
	\in
	\mathcal S_1^+(\Hcal),
	\label{eq:sensing_covariance}
\end{equation}
where \(\bm a\otimes\bm b\) denotes the rank-one operator defined in~\eqref{eq:rank_one_operator}. The covariance operator of the total transmit current is then
\begin{equation}
	\Qcal
	=
	\mathbb E[\bm j\otimes\bm j]
	=
	\sum_{k=1}^{K}
	\bm w_k\otimes\bm w_k
	+
	\Qcal_s
	\in
	\mathcal S_1^+(\Hcal).
	\label{eq:total_covariance}
\end{equation}
Consequently, the average transmit power is given by
\begin{equation}
	\Tr(\Qcal)
	=
	\sum_{k=1}^{K}
	\norm{\bm w_k}_{\Hcal}^2
	+
	\Tr(\Qcal_s),
	\label{eq:total_power}
\end{equation}
where \(\Tr(\cdot)\) denotes the operator trace defined
in~\eqref{eq:tr_def_1}.

Electromagnetic propagation from the transmit aperture to an observation point is characterized by the free-space dyadic Green tensor. For an observation point \(\bm r\in\mathbb R^3\) and a source point \(\bm s\in\Scal_T\), define
\begin{equation}
	R=\norm{\bm r-\bm s}_2,
	\qquad
	\widehat{\bm R}
	=
	\frac{\bm r-\bm s}{R},
	\qquad
	k_0=\frac{2\pi}{\lambda},
	\label{eq:green_geometry}
\end{equation}
where \(\lambda\) denotes the carrier wavelength. The corresponding dyadic Green tensor is
\begin{align}
	\bm{\mathsf G}(\bm r,\bm s)
	=&\;
	\alpha_0\frac{e^{-jk_0R}}{4\pi R}
	\Bigg[
	\left(
	1-\frac{j}{k_0R}-\frac{1}{(k_0R)^2}
	\right)\bm I_3
	\nonumber\\
	&\quad+
	\left(
	-1+\frac{3j}{k_0R}
	+\frac{3}{(k_0R)^2}
	\right)
	\widehat{\bm R}\widehat{\bm R}^{T}
	\Bigg],
	\label{eq:dyadic_green}
\end{align}
where \(\alpha_0\) collects the deterministic electromagnetic scaling constants. Hence, the electric field induced by an aperture current
\(\bm j\in\Hcal\) at \(\bm r\) is
\begin{equation}
	\bm e(\bm r)
	=
	\int_{\Scal_T}
	\bm{\mathsf G}(\bm r,\bm s)
	\bm j(\bm s)\,d\bm s.
	\label{eq:radiated_field}
\end{equation}

For the communication link, user $k$ is located at $\bm r_k\in\mathbb R^3$ and employs a fixed unit-norm polarization combiner $\bm\psi_k\in\mathbb C^3$. After polarization combining, the received signal is
\begin{equation} \label{eq:rx_sig}
	\begin{aligned}
		y_k & = \bm\psi_k^H\bm e(\bm r_k) + n_k= 	\inner{\bm h_k}{\bm j}_{\Hcal} + n_k,\\
		& =
		\inner{\bm h_k}{\bm w_k}_{\Hcal}x_k
		+
		\sum_{i\neq k}
		\inner{\bm h_k}{\bm w_i}_{\Hcal}x_i
		+
		\inner{\bm h_k}{\bm v}_{\Hcal}
		+
		n_k,
	\end{aligned}
\end{equation}
where $\bm h_k(\bm s)
\triangleq 
\bm{\mathsf G}^{H}(\bm r_k,\bm s)
\bm\psi_k \in\Hcal$ denotes the aperture-domain electromagnetic response of user $k$, and $n_k\sim\mathcal{CN}(0,\sigma_k^2)$.

By treating the multiuser interference and the dedicated sensing signal as interference, the receive SINR of user $k$ is
\begin{equation}
	\mathrm{SINR}_k
	=
	\frac{
		\left|
		\inner{\bm h_k}{\bm w_k}_{\Hcal}
		\right|^2
	}{
		\displaystyle
		\sum_{i\neq k}
		\left|
		\inner{\bm h_k}{\bm w_i}_{\Hcal}
		\right|^2
		+
		\inner{\bm h_k}{\Qcal_s\bm h_k}_{\Hcal}
		+
		\sigma_k^2
	}.
	\label{eq:SINR}
\end{equation}

\subsection{Extended-Target Electromagnetic Sensing Model}
\label{subsec:extended_target}

Consider the $q$-th extended sensing target occupying a bounded region $\Omega_q\subset\mathbb R^3$. For a given aperture current $\bm j\in\Hcal$, the average illumination delivered to the $q$-th extended target over $\Omega_q$ is measured by the total electromagnetic-field energy  
\begin{align}
	I_q = 		\mathbb E\!\left[
	P_q(\bm j)
	\right]
	= \mathbb E\!\left[
	\int_{\Omega_q}
	\left\|
	\int_{\Scal_T}
	\bm{\mathsf G}(\bm\xi,\bm s)
	\bm j(\bm s)\,d\bm s
	\right\|_2^2
	d\bm\xi \right].
	\label{eq:instantaneous_illumination}
\end{align}
To express~\eqref{eq:instantaneous_illumination} compactly, we define the target-field Hilbert space 
\begin{equation}
	\Ycal_q
	=
	L^2(\Omega_q;\mathbb C^3)
	\label{eq:target_space}
\end{equation}
and the aperture-to-target operator $\Acal_q:\Hcal\rightarrow\Ycal_q$ as
\begin{equation}
	(\Acal_q\bm j)(\bm\xi)
	=
	\int_{\Scal_T}
	\bm{\mathsf G}(\bm\xi,\bm s)
	\bm j(\bm s)\,d\bm s,
	\qquad
	\bm\xi\in\Omega_q.
	\label{eq:Aq}
\end{equation}
Then, the illumination metric \eqref{eq:instantaneous_illumination} can be written as
\begin{equation}
	\begin{aligned}
		I_q
		&=
		\mathbb E\!\left[
		\norm{\Acal_q\bm j}_{\Ycal_q}^2
		\right]
		=
		\mathbb E\!\left[
		\inner{\bm j}{\Rcal_q\bm j}_{\Hcal}
		\right]
		=
		\Tr(\Rcal_q\Qcal),
	\end{aligned}
	\label{eq:illumination_equivalent_forms}
\end{equation}
where $
\Rcal_q
\triangleq
\Acal_q^*\Acal_q
:
\Hcal\rightarrow\Hcal,
$
denotes  the target illumination operator. The last equality in~\eqref{eq:illumination_equivalent_forms} follows from
\begin{equation}\label{eq:rankone_trace_identity}
	\Tr\!\left(
	\Rcal_q(\bm j\otimes\bm j)
	\right)
	=
	\inner{\bm j}{\Rcal_q\bm j}_{\Hcal},
\end{equation}
together with \(\Qcal=\mathbb E[\bm j\otimes\bm j]\) defined in~\eqref{eq:total_covariance}.

Let \(\omega_q\ge0\) denote the sensing weight assigned to target \(q\). The weighted illumination utility
\(
U_{\rm ill}\triangleq\sum_{q=1}^{Q}\omega_q I_q
\)
can be written as
\begin{equation}
	U_{\rm ill}
	=
	\Tr(\Rcal\Qcal)
	=
	\sum_{k=1}^{K}
	\inner{\bm w_k}{\Rcal\bm w_k}_{\Hcal}
	+
	\Tr(\Rcal\Qcal_s),
	\label{eq:total_illumination}
\end{equation}
where
\(
\Rcal=\sum_{q=1}^{Q}\omega_q\Rcal_q
\)
is the aggregate sensing operator.

\vspace{-5pt}
\subsection{Infinite-Dimensional Beamforming Formulation}
\label{subsec:problem_formulation}

We jointly design the information beamformers $\{\bm w_k\}_{k=1}^{K}$ and the sensing covariance $\Qcal_s$ to maximize the weighted target illumination while satisfying the communication QoS and transmit-power requirements. The resulting problem is
\begin{subequations}
	\label{prob:P0}
	\begin{align}
		\mathcal P_0:\quad
		\max_{\{\bm w_k\},\,\Qcal_s}\quad
		&
		\Tr(\Rcal\Qcal)
		\label{prob:P0_obj}
		\\
		\mathrm{s.t.}\quad
		&
		\mathrm{SINR}_k\ge\gamma_k,
		\qquad
		k=1,\ldots,K,
		\label{prob:P0_sinr}
		\\
		&
		\sum_{k=1}^{K}
		\norm{\bm w_k}_{\Hcal}^2
		+
		\Tr(\Qcal_s)
		\le
		P_{\max},
		\label{prob:P0_power}
		\\
		&
		\bm w_k\in\Hcal,
		\quad k=1,\ldots,K,\\
		&
		\Qcal_s\in\mathcal S_1^+(\Hcal).
		\label{prob:P0_domain}
	\end{align}
\end{subequations}
Problem~\eqref{prob:P0} is an infinite-dimensional nonconvex optimization problem due to the functional beamforming variables and SINR constraints. Moreover, the extended-target sensing operator \(\Rcal\) may have infinite rank, making it difficult to identify a priori the optimal solution subspace. In the next section, we will develop a tractable framework that addresses both issues and enables global optimization of $\mathcal P_0$.

\vspace{-5pt}
\section{An Operator SDR Approach to ${\cal P}_0$}
\label{sec:method}
In this section, we derive an operator SDR of $\mathcal P_0$ and establish its optimal covariance structure and tightness.

For each communication user, define the rank-one information covariance operator
\begin{equation}
	\Qcal_k
	=
	\bm w_k\otimes\bm w_k
	\in
	\mathcal S_1^+(\Hcal),
	\qquad
	k=1,\ldots,K.
	\label{eq:Qk}
\end{equation}
For a nonzero positive operator $\Qcal_k$, the factorization in \eqref{eq:Qk} is equivalent to
\begin{equation}
	\Qcal_k\succeq0,
	\qquad
	\rank(\Qcal_k)=1.
	\label{eq:rankone_equiv}
\end{equation}
Rewriting the quadratic terms in $\mathcal P_0$ using $\{\Qcal_k\}$ and dropping the rank-one constraints yields the operator SDR:
\begin{subequations}
	\label{prob:sdr}
	\begin{align}
		\mathcal P_{\rm SDR}:
		\max_{\{\Qcal_k\},\,\Qcal_s}
		\quad&
		\Tr(\Rcal\Qcal)
		\label{prob:sdr_obj}
		\\
		\mathrm{s.t.}\quad
		&
		(1+\gamma_k)
		\Tr(\Hcal_k\Qcal_k)
		-
		\gamma_k
		\Tr(\Hcal_k\Qcal)
		\nonumber\\
		&\hspace{12mm}
		\ge
		\gamma_k\sigma_k^2,
		\quad
		k=1,\ldots,K,
		\label{prob:sdr_sinr}
		\\
		&
		\Tr(\Qcal)\le P_{\max},
		\label{prob:sdr_power}
		\\
		&
		\Qcal_k,\Qcal_s
		\in
		\mathcal S_1^+(\Hcal),
		\label{prob:sdr_domain}
	\end{align}
\end{subequations}
where $\Hcal_i
\triangleq
\bm h_i\otimes\bm h_i$ and \eqref{prob:sdr_sinr} follows from $\left|
\inner{\bm h_i}{\bm w_k}_{\Hcal}
\right|^2
=
\Tr(\Hcal_i\Qcal_k)$.
Problem~\eqref{prob:sdr} is convex but remains infinite dimensional. Unlike conventional finite-dimensional matrix optimization problems, the existence of an optimizer in infinite-dimensional spaces is not guaranteed without additional compactness properties. Therefore, before characterizing its optimal structure, we first establish the existence of an optimal solution.
\begin{lemma}
	\label{lem:sdr_existence}
	Suppose that \eqref{prob:sdr} is feasible. Then, an optimal solution of
	\eqref{prob:sdr} exists.
\end{lemma}

{\it Proof.}~See Appendix~\ref{app:sdr_existence} \hfill $\blacksquare$

With Lemma~\ref{lem:sdr_existence}, we next characterize the finite-rank structure of the optimal covariance operators. We distinguish between the cases $\rank(\Rcal)<\infty$ and $\rank(\Rcal)=\infty$.

\vspace{-10pt}
\subsection{$\rank(\Rcal)<\infty$}
\label{subsubsec:finite_rank_R}

A representative finite-rank sensing scenario arises when each target is modeled by a finite number of point scatterers. In this case, the electromagnetic response is characterized by a finite-dimensional set of field components, and the resulting target illumination operators $\Rcal_q$, and hence their aggregate $\Rcal$, have finite rank.

Define the total response subspace
\begin{equation}
	\Ucal_{\rm fr}
	\triangleq
	\spann\!\left(
	\ran(\Rcal)
	\cup
	\{\bm h_1,\ldots,\bm h_K\}
	\right).
	\label{eq:finite_response_space}
\end{equation}
Since $\rank(\Rcal)<\infty$ and $K<\infty$,  
$\Ucal_{\rm fr}$ is a finite-dimensional closed subspace of $\Hcal$. Both the sensing objective and the communication constraints depend only on the components of the covariance operators within $\Ucal_{\rm fr}$. Accordingly, the transmit covariance operators can be restricted to this subspace without loss of optimality.

\begin{theorem}
	\label{thm:finite_rank_response}
	If $\rank(\Rcal)<\infty$, then $\mathcal P_{\rm SDR}$ in
	\eqref{prob:sdr} admits an optimal solution
	$(\{\Qcal_k^\star\}_{k=1}^{K},\Qcal_s^\star)$ such that
	\begin{equation}
		\ran(\Qcal_j^\star)
		\subseteq
		\Ucal_{\rm fr},
		\qquad
		j\in\{s,1,\ldots,K\}.
		\label{eq:finite_response_range}
	\end{equation}
	Consequently,
	\begin{equation}
		\rank(\Qcal_j^\star)
		\le
		\dim(\Ucal_{\rm fr})
		<
		\infty,
		\qquad
		j\in\{s,1,\ldots,K\}.
		\label{eq:finite_response_rank_bound}
	\end{equation}
\end{theorem}

{\it Proof.}~See Appendix~\ref{app:finite_rank_response}. \hfill $\blacksquare$

Theorem~\ref{thm:finite_rank_response} establishes a finite-rank optimal solution of $\mathcal P_{\rm SDR}$ supported on the finite-dimensional subspace $\Ucal_{\rm fr}$. This finite-dimensional optimal structure enables an equivalent matrix SDP reformulation of the operator problem. We defer the  matrix reformulation to Section~\ref{sec:subspace_search} for a detailed discussion.

\subsection{$\rank(\Rcal)=\infty$}
\label{subsubsec:infinite_rank_R}

A representative infinite-rank sensing scenario arises when an extended target is modeled as a continuously distributed spatial region rather than a finite collection of point scatterers. In this case, the electromagnetic response over the target region may involve infinitely many spatial modes, so the illumination operators $\Rcal_q$ and their aggregate $\Rcal$ may have infinite rank.
The response subspace $\Ucal_{\rm fr}$ is therefore no longer finite dimensional. Nevertheless, the optimal covariance operators may still exhibit an intrinsic finite-rank structure, as shown below.
\begin{theorem}
	\label{thm:finite_rank_sdr}
	Assume constraint qualification holds for $\mathcal P_{\rm SDR}$. If $\rank(\Rcal)=\infty$, then every optimal solution
	$(\{\Qcal_k^\star\}_{k=1}^{K},\Qcal_s^\star)$ satisfies
	\begin{equation}
		\rank(\Qcal_j^\star)
		<\infty,
		\qquad
		j\in\{s,1,\ldots,K\}.
		\label{eq:infinite_sensing_finite_rank}
	\end{equation}
\end{theorem}

{\it Proof.}~See Appendix~\ref{app:finite_rank_theorem}. \hfill $\blacksquare$

Theorems~\ref{thm:finite_rank_response} and~\ref{thm:finite_rank_sdr} highlight a fundamental distinction between the finite- and infinite-rank sensing regimes. When $\rank(\Rcal)<\infty$, an exact finite-dimensional optimal support is available a priori through a finite basis of the response subspace $\Ucal_{\rm fr}$. In contrast, when $\rank(\Rcal)=\infty$, although the optimal covariance operators remain finite rank, the electromagnetic modes supporting them are not known in advance, since the corresponding response space is infinite dimensional. Consequently, unlike the finite-rank case, $\mathcal P_{\rm SDR}$ cannot be directly reduced to a finite-dimensional matrix SDP using a prescribed finite basis.

Thus far, we have shown the finite-rank structure of optimal solutions to $\mathcal P_{\rm SDR}$. It remains to establish whether an optimal solution of $\mathcal P_0$ can be recovered from the relaxed problem. The following result proves the tightness of the operator SDR and provides a constructive rank-one recovery.
\begin{theorem}
	\label{thm:sdr_tight}
	Suppose that $\mathcal P_{\rm SDR}$ in \eqref{prob:sdr} is feasible, and let $(\{\Qcal_k^\star\}_{k=1}^{K},\Qcal_s^\star)$ be any optimal solution of $\mathcal P_{\rm SDR}$. 
	Let
	\begin{equation}
		\widehat{\bm w}_k
		=
		\frac{\Qcal_k^\star\bm h_k}{\sqrt{\Tr(\Hcal_k\Qcal_k^\star)}},
		\qquad
		\widehat{\Qcal}_k
		=
		\widehat{\bm w}_k\otimes\widehat{\bm w}_k,~~\forall~k
		\label{eq:rank_recovery}
	\end{equation}
	and 
	\begin{equation}
		\widehat{\Qcal}_s
		=
		\Qcal_s^\star
		+
		\sum_{k=1}^{K}
		\left(
		\Qcal_k^\star-\widehat{\Qcal}_k
		\right).
		\label{eq:Qs_recovery}
	\end{equation}
	Then $\widehat{\Qcal}_s\in\mathcal S_1^+(\Hcal)$ and $\rank(\widehat{\Qcal}_k)=1$ for all $k$. Moreover, $\{\widehat{\bm w}_k\}_{k=1}^{K}$ together with $\widehat{\Qcal}_s$ is globally optimal for $\mathcal P_0$. 
\end{theorem}

{\it Proof.}~See Appendix~\ref{app:sdr_tightness}. \hfill $\blacksquare$

\section{Adaptive Electromagnetic-Mode Enrichment Method for ${\cal P}_{\rm SDR}$}
\label{sec:mode_enrichment}

This section focuses on solving \(\mathcal P_{\rm SDR}\), particularly for the infinite-rank case. The proposed method adopts an adaptive subspace-search strategy. Specifically, it starts from a low-dimensional electromagnetic subspace and iteratively enriches this subspace according to the optimality information revealed by the current restricted solution. In this way, the method progressively identifies the electromagnetic modes supporting the global optimum.

\vspace{-10pt}

\subsection{Restricted Electromagnetic Search Subspace} \label{sec:subspace_search}

At iteration \(m\), the adaptive procedure maintains a finite-dimensional search subspace \(\Ucal_m\subset\Hcal\). Let \(r_m=\dim(\Ucal_m)\) and \(\{\bm u_1,\ldots,\bm u_{r_m}\}\) be an orthonormal basis of \(\Ucal_m\). Define the associated isometric basis operator
\begin{equation}
	\bm\Psi_m:\mathbb C^{r_m}\rightarrow\Hcal, 	\qquad
	\bm\Psi_m\bm a
	=
	\sum_{n=1}^{r_m}
	a_n\bm u_n,
	\label{eq:Psi_m}
\end{equation}
which satisfies $\ran(\bm\Psi_m)=\Ucal_m,\bm\Psi_m^*\bm\Psi_m=\bm I_{r_m}$.
Restrict the covariance operators to $\Ucal_m$ as
\begin{equation}
	\Qcal_k
	=
	\bm\Psi_m\bm Q_{k,m}\bm\Psi_m^*,
	\qquad
	\Qcal_s
	=
	\bm\Psi_m\bm Q_{s,m}\bm\Psi_m^*,
	\label{eq:restricted_covariances}
\end{equation}
where $\bm Q_{k,m},\ \bm Q_{s,m} \in \mathbb{C}^{r_m\times r_m}$ are positive semidefinite matrices. 
Substituting \eqref{eq:restricted_covariances} into \eqref{prob:sdr} yields the finite-dimensional SDP restricted to \(\Ucal_m\):
\begin{subequations}
	\label{prob:Pm}
	\begin{align}
		\mathcal P_m:\ 
		\max_{\{\bm Q_{k,m}\},\bm Q_{s,m}}
		&
		\Tr(\bm R_m\bm Q_m)
		\label{prob:Pm_obj}
		\\
		\mathrm{s.t.}\quad
		&
		(1+\gamma_k)
		\Tr(\bm H_{k,m}\bm Q_{k,m})
		\nonumber\\[-1mm]
		&\,\,
		-\gamma_k
		\Tr(\bm H_{k,m}\bm Q_m)
		\ge
		\gamma_k\sigma_k^2,
		\,\, \forall k,
		\label{prob:Pm_sinr}
		\\
		&
		\Tr(\bm Q_m)\le P_{\max},
		\label{prob:Pm_power}
		\\
		&	
		\bm Q_{k,m}\succeq0,~\forall k, \qquad \bm Q_{s,m}\succeq0,
		\label{prob:Pm_psd}
	\end{align}
\end{subequations}
where $	\bm R_m
\triangleq 
\bm\Psi_m^*\Rcal\bm\Psi_m$, $\bm H_{k,m}
\triangleq
\bm\Psi_m^*\Hcal_k\bm\Psi_m$ and $\bm Q_m \triangleq \sum_{k=1}^{K}\bm Q_{k,m} + \bm Q_{s,m}$.

Problem~\eqref{prob:Pm} is a standard finite-dimensional SDP and can be solved optimally. Its solution, however, is generally optimal only over the current subspace $\Ucal_m$, rather than for the full-space problem~\eqref{prob:sdr}. This raises two key questions: whether $\Ucal_m$ is sufficient to achieve full-space optimality, and, if not, how the search space should be enriched to capture an optimal solution of~\eqref{prob:sdr}. These issues are addressed in the next subsection.

\begin{remark}
	When $\rank(\Rcal)<\infty$, $\bm\Psi_m$ can be chosen as an isometric basis operator of $\Ucal_{\rm fr}$. By Theorem~\ref{thm:finite_rank_response}, there exists an optimal solution of~\eqref{prob:sdr} whose covariance operators admit the representation in~\eqref{eq:restricted_covariances}. Consequently, problem~\eqref{prob:Pm} provides an exact finite-dimensional reformulation of~\eqref{prob:sdr}, and an optimal full-space solution can be recovered from the optimal matrix solution of~\eqref{prob:Pm} through \eqref{eq:restricted_covariances}.
\end{remark}


\subsection{Restricted and Full-Space Optimality Conditions}
\label{subsec:optimality_conditions}

We next characterize whether the current search subspace $\Ucal_m$ is sufficient to attain full-space optimality. These conditions will also provide the basis for the subsequent search-space enrichment. Suppose that $\mathcal P_m$ satisfies Slater's condition, and let
$
(
\{\bm Q_{k,m}^{\star}\}_{k=1}^{K},
\bm Q_{s,m}^{\star},
\{\lambda_{k,m}\}_{k=1}^{K},
\mu_m
)
$
be a primal--dual optimal Karush–Kuhn–Tucker (KKT) point of~\eqref{prob:Pm}, where $\lambda_{k,m}\ge0$ and $\mu_m\ge0$ are the dual multipliers associated with the SINR constraints and the transmit-power constraint, respectively.
The KKT conditions of $\mathcal P_m$ are 
\begin{subnumcases}{\label{eq:restricted_KKT}}
	\begin{aligned}
		&(1+\gamma_k)
		\Tr(\bm H_{k,m}\bm Q_{k,m}^{\star})
		-
		\gamma_k
		\Tr(\bm H_{k,m}\bm Q_m^{\star})
		\\[-1mm]
		&\hspace{1.2cm}
		\ge
		\gamma_k\sigma_k^2,
		\qquad  k=1,\ldots,K,
	\end{aligned}
	\label{eq:restricted_KKT_primal_sinr}
	\\
	\Tr(\bm Q_m^\star)\le P_{\max},
	~
	\bm Q_{j,m}^{\star}\succeq\bm 0,
	~ j\in\{s,1,\ldots,K\},
	\label{eq:restricted_KKT_primal_other}
	\\
	\mu_m\ge0,~~\lambda_{k,m}\ge0,
	\quad k=1,\ldots,K,
	\label{eq:restricted_KKT_scalar_dual}
	\\
	\bm Z_{j,m}^{(r)}\succeq\bm 0,
	\qquad
	j\in\{s,1,\ldots,K\},
	\label{eq:restricted_KKT_dual_slack}
	\\
	\begin{aligned}
		\lambda_{k,m}
		\Big[
		&(1+\gamma_k)
		\Tr(\bm H_{k,m}\bm Q_{k,m}^{\star})
		-
		\gamma_k
		\Tr(\bm H_{k,m}\bm Q_m^{\star})
		\\[-1mm]
		&-
		\gamma_k\sigma_k^2
		\Big]
		=0,
		\qquad k=1,\ldots,K,
	\end{aligned}
	\label{eq:restricted_KKT_comp_sinr}
	\\
	\mu_m
	\left(
	P_{\max}
	-
	\Tr(\bm Q_m^\star)
	\right)
	=0,
	\label{eq:restricted_KKT_comp_power}
	\\
	\Tr\!\left(
	\bm Z_{j,m}^{(r)}
	\bm Q_{j,m}^{\star}
	\right)
	=0,
	\qquad
	j\in\{s,1,\ldots,K\},
	\label{eq:restricted_KKT_comp_matrix}
\end{subnumcases}
where 
\begingroup
\small
\begin{align}
	&\bm Z_{j,m}^{(r)}
	\triangleq
	\mu_m\bm I_{r_m}
	-
	\bm B_{j,m}^{(r)},
	\qquad
	j\in\{s,1,\ldots,K\},
	&&
	\\
	&\bm B_{s,m}^{(r)}
	\triangleq
	\bm R_m
	-
	\sum_{k=1}^{K}
	\gamma_k\lambda_{k,m}\bm H_{k,m},
	&&
	\\
	&\bm B_{i,m}^{(r)}
	\triangleq
	\bm R_m
	\!+\!
	\lambda_{i,m}\bm H_{i,m}
	\!-\!
	\sum_{k\neq i}
	\gamma_k\lambda_{k,m}\bm H_{k,m},
	\, i=1,\ldots,K.
\end{align}
\endgroup

We next consider the KKT conditions of the full-space SDR~\eqref{prob:sdr}. Lift the restricted covariance solution to $\Hcal$ as 
\begin{equation}
	\widehat{\Qcal}_{j,m}
	=
	\bm\Psi_m
	\bm Q_{j,m}^{\star}
	\bm\Psi_m^*,
	\qquad
	j\in\{s,1,\ldots,K\},
	\label{eq:lifted_restricted_solution}
\end{equation}
and define 
\begin{equation}
	\widehat{\Qcal}_m
	=
	\sum_{k=1}^{K}\widehat{\Qcal}_{k,m}
	+
	\widehat{\Qcal}_{s,m}.
	\label{eq:lifted_total_covariance}
\end{equation}
Using the same scalar dual multipliers $\{\lambda_{k,m}\}_{k=1}^{K}$ and $\mu_m$, the  full-space KKT conditions of problem~\eqref{prob:sdr} are 
{\small\begin{subnumcases}{\label{eq:full_KKT}}
	\begin{aligned}
		&(1+\gamma_k)
		\Tr(\Hcal_k\widehat{\Qcal}_{k,m})
		-
		\gamma_k
		\Tr(\Hcal_k\widehat{\Qcal}_m)
		\\[-1mm]
		&\hspace{1.2cm}
		\ge
		\gamma_k\sigma_k^2,
		\qquad k=1,\ldots,K,
	\end{aligned}
	\label{eq:full_KKT_primal_sinr}
	\\
	\Tr(\widehat{\Qcal}_m)\le P_{\max},
	~
	\widehat{\Qcal}_{j,m}\succeq 0,
	~ j\in\{s,1,\ldots,K\},
	\label{eq:full_KKT_primal_other}
	\\
	\mu_m\ge0,~~\lambda_{k,m}\ge0,
	\quad k=1,\ldots,K,
	\label{eq:full_KKT_scalar_dual}
	\\
	\Zcal_{j,m}\succeq 0,
	\qquad
	j\in\{s,1,\ldots,K\},
	\label{eq:full_KKT_dual_slack}
	\\
	\begin{aligned}
		\lambda_{k,m}
		\Big[
		&(1+\gamma_k)
		\Tr(\Hcal_k\widehat{\Qcal}_{k,m})
		-
		\gamma_k
		\Tr(\Hcal_k\widehat{\Qcal}_m)
		\\[-1mm]
		&-
		\gamma_k\sigma_k^2
		\Big]
		=0,
		\qquad k=1,\ldots,K,
	\end{aligned}
	\label{eq:full_KKT_comp_sinr}
	\\
	\mu_m
	\left(
	P_{\max}
	-
	\Tr(\widehat{\Qcal}_m)
	\right)
	=0,
	\label{eq:full_KKT_comp_power}
	\\
	\Tr\!\left(
	\Zcal_{j,m}
	\widehat{\Qcal}_{j,m}
	\right)
	=0,
	\qquad
	j\in\{s,1,\ldots,K\},
	\label{eq:full_KKT_comp_operator}
\end{subnumcases}}
where 
\begin{align}
	\Zcal_{j,m}
	& \triangleq
	\mu_m\Ical
	-
	\Bcal_{j,m},
	\qquad
	j\in\{s,1,\ldots,K\}.
	\label{eq:full_dual_slack_def} \\
	\Bcal_{s,m}
	& \triangleq
	\Rcal
	-
	\sum_{k=1}^{K}
	\gamma_k\lambda_{k,m}\Hcal_k,
	\label{eq:Bs}
	\\
	\Bcal_{i,m}
	&\triangleq
	\Rcal
	+
	\lambda_{i,m}\Hcal_i
	-
	\sum_{k\neq i}
	\gamma_k\lambda_{k,m}\Hcal_k,
	~
	i=1,\ldots,K,
	\label{eq:Bi}
\end{align}
The relationship between the two KKT systems is summarized as follows.

\begin{Fact}\label{fact:restricted_full_kkt}
	The lifted primal--dual optimal KKT point of $\mathcal P_m$ satisfies \eqref{eq:full_KKT_primal_sinr}--\eqref{eq:full_KKT_scalar_dual} and \eqref{eq:full_KKT_comp_sinr}--\eqref{eq:full_KKT_comp_operator}. The only condition that is not guaranteed is the full-space dual-slack feasibility condition
	\eqref{eq:full_KKT_dual_slack}.
\end{Fact}

{\it Proof.}~See Appendix~\ref{app:restricted_full_kkt}. \hfill $\blacksquare$

Fact~\ref{fact:restricted_full_kkt} reveals the key distinction between restricted and full-space optimality. Specifically, the restricted dual-slack matrix satisfies
\begin{equation}
	\bm Z_{j,m}^{(r)}
	=
	\bm\Psi_m^*
	\Zcal_{j,m}
	\bm\Psi_m
	\succeq0,
	\label{eq:restricted_full_slack_relation}
\end{equation}
which guarantees
\begin{equation}
	\inner{\bm u}
	{\Zcal_{j,m}\bm u}_{\Hcal}
	\ge0,
	\qquad
	\forall\,\bm u\in\Ucal_m.
	\label{eq:restricted_slack_on_Um}
\end{equation}
Full-space dual feasibility \eqref{eq:full_KKT_dual_slack}, however, requires
\begin{equation}
	\inner{\bm u}
	{\Zcal_{j,m}\bm u}_{\Hcal}
	\ge0,
	\qquad
	\forall\,\bm u\in\Hcal.
	\label{eq:full_slack_on_H}
\end{equation}
Hence, if the dual variables obtained from the restricted problem~\eqref{prob:Pm} also satisfy~\eqref{eq:full_slack_on_H}, then, by Fact~\ref{fact:restricted_full_kkt}, the lifted restricted solution satisfies the full-space KKT conditions and is globally optimal for~\eqref{prob:sdr}. Otherwise, the current search subspace $\Ucal_m$
is insufficient to certify full-space optimality, and there exists at least one block $j$ for which the dual-slack operator may become negative along a direction outside $\Ucal_m$, i.e., there may exist a unit-norm $\bm u\in\Hcal\setminus\Ucal_m$ such that
\begin{equation}
	\inner{\bm u}
	{\Zcal_{j,m}\bm u}_{\Hcal}
	<0.
	\label{eq:full_slack_violation}
\end{equation}

This motivates quantifying the departure from full-space dual feasibility. Since $\Rcal$ is compact and each $\Hcal_k$ is rank one, $\Bcal_{j,m}$ is compact and self-adjoint. By the spectral theorem for compact self-adjoint operators~\cite{kreyszig1978introductory}, its largest Rayleigh quotient is finite and attained whenever it is positive. Define
\begin{equation}
	\rho_{j,m}
	\triangleq
	\sup_{\|\bm u\|_{\Hcal}=1}
	\inner{\bm u}
	{\Bcal_{j,m}\bm u}_{\Hcal}.
	\label{eq:rho_j}
\end{equation}
Because $\Zcal_{j,m}=\mu_m\Ical-\Bcal_{j,m}$, the full-space dual-feasibility condition is equivalent to
\begin{equation}
	\rho_{j,m}\le\mu_m.
	\label{eq:rho_dual_feasibility}
\end{equation}
Accordingly, define the block-wise dual-slack violation as
\begin{equation}
	\Delta_{j,m}
	\triangleq
	\rho_{j,m}-\mu_m,
	\qquad
	j\in\{s,1,\ldots,K\},
	\label{eq:Delta_jm}
\end{equation}
and the maximum violation
\begin{equation}
	\Delta_m
	\triangleq
	\max_{j\in\{s,1,\ldots,K\}}
	\Delta_{j,m}.
	\label{eq:Delta_m}
\end{equation}
Hence, $\Delta_m\le0$ certifies full-space dual feasibility and, by Fact~\ref{fact:restricted_full_kkt}, full-space optimality of the lifted restricted solution. In contrast, $\Delta_m>0$ indicates a violation of the full-space dual-slack conditions.  
Moreover, the positive part of $\Delta_m$ provides a computable bound on the remaining optimality gap.
\begin{proposition}
	\label{prop:gap}
	Let $p_m$ and $p^\star$ denote the optimal values of \eqref{prob:Pm} and \eqref{prob:sdr}, respectively, and define
	\begin{equation}
		\delta_m
		\triangleq
		[\Delta_m]_+
		=
		\max\{\Delta_m,0\}.
		\label{eq:delta_m}
	\end{equation}
	Then
	\begin{equation}
		0
		\le
		p^\star-p_m
		\le
		\delta_mP_{\max}.
		\label{eq:gap_bound}
	\end{equation}
	In particular, if $\Delta_m\le0$, then
	\begin{equation}
		p_m=p^\star,
		\label{eq:global_certificate}
	\end{equation}
	and the lifted solution in
	\eqref{eq:lifted_restricted_solution} is globally optimal for
	\eqref{prob:sdr}.
\end{proposition}

{\it Proof.}~See Appendix~\ref{app:gap_bound}. \hfill $\blacksquare$

Proposition~\ref{prop:gap} provides a criterion for evaluating the sufficiency of the current search subspace. When $\Delta_m>0$, the current subspace fails to certify full-space optimality, and additional electromagnetic modes are required. 
The next subsection develops the corresponding adaptive enrichment strategy.

\subsection{Missing-Mode Identification and Space Enrichment}
\label{subsec:mode_enrichment}

When $\Delta_m>0$, select the most violated covariance block as
\begin{equation}
	j_m^\star
	\in
	\argmax_{j\in\{s,1,\ldots,K\}}
	\left(
	\rho_{j,m}-\mu_m
	\right),
	\label{eq:jstar}
\end{equation}
which gives
\begin{equation}
	\rho_{j_m^\star,m}
	=
	\mu_m+\Delta_m
	>0.
	\label{eq:positive_rho_jstar}
\end{equation}
Since $\Bcal_{j_m^\star,m}$ is compact and self-adjoint, its positive Rayleigh supremum is attained by a unit-norm eigenfunction $\bm v_m\in\Hcal$~\cite{kreyszig1978introductory}, satisfying
\begin{equation}
	\Bcal_{j_m^\star,m}\bm v_m
	=
	\rho_{j_m^\star,m}\bm v_m .
	\label{eq:violating_eigenmode}
\end{equation}

The identified eigenmode corresponds to a violating direction of the full-space dual-slack operator. Specifically,
\begin{equation}
	\inner{\bm v_m}
	{\Zcal_{j_m^\star,m}\bm v_m}_{\Hcal}
	=
	\mu_m-\rho_{j_m^\star,m}
	=
	-\Delta_m
	<0 .
	\label{eq:negative_direction}
\end{equation}
On the other hand, the restricted KKT condition guarantees
\begin{equation}
	\inner{\bm u}
	{\Zcal_{j_m^\star,m}\bm u}_{\Hcal}
	\ge0,
	\qquad
	\forall\,\bm u\in\Ucal_m .
	\label{eq:positive_on_Um}
\end{equation}
Thus, $\bm v_m\notin\Ucal_m$ and $\bm v_m$ identifies an electromagnetic direction absent from the current search subspace. Let
\begin{equation}
	\bm \Pi_m
	=
	\bm\Psi_m\bm\Psi_m^*
	\label{eq:projection}
\end{equation}
denote the orthogonal projector onto $\Ucal_m$. The component of $\bm v_m$ orthogonal to $\Ucal_m$ is given by
\begin{equation}
	\widetilde{\bm u}_{m+1}
	=
	(\Ical-\bm\Pi_m)\bm v_m ,
	\label{eq:newmode_raw}
\end{equation}
and the normalized enrichment mode is
\begin{equation}
	\bm u_{m+1}
	=
	\frac{
		\widetilde{\bm u}_{m+1}
	}
	{
		\norm{\widetilde{\bm u}_{m+1}}_{\Hcal}
	}.
	\label{eq:newmode}
\end{equation}
The search subspace is then updated as
\begin{equation}
	\Ucal_{m+1}
	=
	\Ucal_m
	\oplus
	\spann\{\bm u_{m+1}\},
	\label{eq:space_update}
\end{equation}
which guarantees 
\begin{equation} \label{eq:nest_U}
	\Ucal_m\subseteq\Ucal_{m+1}.
\end{equation}
The resulting adaptive electromagnetic-mode enrichment procedure is summarized in Algorithm~\ref{alg:mode_enrichment}.
\begin{algorithm}[t]
	\caption{Dual-Slack-Driven Electromagnetic-Mode Enrichment}
	\label{alg:mode_enrichment}
	\begin{algorithmic}[1]
		
		\State Initialize $\Ucal_0=\Ucal_c$ and construct an isometry
		$\bm\Psi_0$ for $\Ucal_0$.
		
		\State Set $m=0$ and choose the tolerance $\epsilon_{\rm sp}\ge0$.
		
		\Repeat
		
		\State Construct  $\bm R_m$ and $\{\bm H_{k,m}\}_{k=1}^{K}$. 
		
		\State Solve $\eqref{prob:Pm}$ to obtain
		$\{\bm Q_{k,m}^{\star}\}$, $\bm Q_{s,m}^{\star}$,
		$\{\lambda_{k,m}\}$, and $\mu_m$.
		
		\State Construct the full-space operators
		$\{\Bcal_{j,m}\}_{j\in\{s,1,\ldots,K\}}$
		using \eqref{eq:Bs}--\eqref{eq:Bi}.
		
		\State Evaluate $\{\rho_{j,m}\}$ and $\Delta_m$.
		
		\If{$\Delta_m>\epsilon_{\rm sp}$}
	
		\State Select the most violated block $j_m^\star$ and compute the
		dominant eigenfunction $\bm v_m$ of
		$\Bcal_{j_m^\star,m}$.
		
		\State Generate the normalized enrichment mode
		$\bm u_{m+1}$ using \eqref{eq:newmode}.
		
		\State Update $\Ucal_{m+1}$ using \eqref{eq:space_update}, construct $\bm\Psi_{m+1}$, and set $m\leftarrow m+1$.  
		
		\EndIf
		
		\Until{$\Delta_m\le\epsilon_{\rm sp}$}
		
		\State Lift the final restricted covariance solution to $\Hcal$
		and apply the rank-one recovery in
		Theorem~\ref{thm:sdr_tight}.
		
	\end{algorithmic}
\end{algorithm}

The nested property in~\eqref{eq:nest_U} guarantees that the sequence $\{p_m\}$ is monotonically nondecreasing, i.e., $p_m \le p_{m+1}$. However, monotonicity alone does
not ensure convergence to the full-space optimum, i.e., $p_m\rightarrow p^\star$ as $m\rightarrow\infty$. To establish the global convergence of the proposed enrichment procedure, we introduce the following regularity assumption.

\begin{assumption}
	\label{ass:dual_bounded}
	Each restricted SDP~\eqref{prob:Pm} admits an optimal primal--dual
	KKT pair. Moreover, the corresponding optimal dual multipliers are
	uniformly bounded, i.e., there exists a constant $C<\infty$ such that
	\begin{equation}
		\mu_m+\sum_{k=1}^{K}\lambda_{k,m}
		\le C,
		\qquad \forall m.
		\label{eq:dual_bounded}
	\end{equation}
\end{assumption}

Assumption~\ref{ass:dual_bounded} imposes a  uniform primal--dual regularity condition on the sequence of restricted SDPs. 
For example, such boundedness can be ensured under a uniform Slater condition. We have the following global convergence result. 

\begin{theorem}
	\label{thm:global_convergence}
	Under Assumption~\ref{ass:dual_bounded}, the mode enrichment sequence generated by Algorithm~\ref{alg:mode_enrichment} satisfies
	\begin{equation}
		p_m \uparrow p^\star,
		\label{eq:pm_to_pinf}
	\end{equation}
	where $p^\star$ is the optimal value of the full-space operator SDR~\eqref{prob:sdr}. 
\end{theorem}

{\it Proof.}~See Appendix~\ref{app:global_convergence}. \hfill $\blacksquare$

Together with the SDR tightness result in Theorem~\ref{thm:sdr_tight}, Theorem~\ref{thm:global_convergence} shows that the adaptive enrichment procedure converges to the global optimum of the original infinite-dimensional beamforming problem~\eqref{prob:P0}.

\section{Simulation Results}
\label{sec:simulation}

This section presents numerical results to evaluate the structural properties of the optimized transmit covariance and the performance of the proposed adaptive electromagnetic-mode enrichment. 

\vspace{-5pt}
\subsection{Simulation Setup and Numerical Implementation}
\label{subsec:simulation_setup}

We consider a square CAPA on the $x$--$y$ plane operating at $f_c=3$~GHz, with wavelength $\lambda=0.1$~m and aperture size $L_x=L_y=0.8$~m, i.e., $8\lambda\times8\lambda$.
Two communication users are located at $(-1.05,0.38,4.50)\,\mathrm m$ and $(0.85,-0.50,4.50)\,\mathrm m$, with fixed polarization combiners
\begin{equation}
	\bm\psi_1
	=
	\frac{[1,\,0.25,\,0]^{T}}{\sqrt{1+0.25^2}},
	\qquad
	\bm\psi_2
	=
	\frac{[0.25,\,1,\,0]^{T}}{\sqrt{1+0.25^2}}.
\end{equation}
Unless otherwise specified, both users require an SINR of $16$~dB, with $\sigma_k^2=10^{-2}$ and $P_{\max}=1$.
Two extended sensing targets lie in the plane $z=2.3$~m, with supports
\begin{align}
	\Omega_1 &: x\in[-1.20,-0.15],\quad
	y\in[-0.05,0.75], \nonumber\\
	\Omega_2 &: x\in[0.25,0.95],\quad
	y\in[-0.70,0.05],
\end{align}
and equal sensing weights $\omega_1=\omega_2=1$. Each communication response is normalized to unit norm, and the sensing operator is normalized by its largest eigenvalue.

All continuous-domain integrals are evaluated using tensor-product Gauss--Legendre quadrature. Let $\{s_{x,i},w_{x,i}^{\rm T}\}_{i=1}^{M_{\rm T}}$ and $\{s_{y,j},w_{y,j}^{\rm T}\}_{j=1}^{M_{\rm T}}$ denote the aperture quadrature nodes and weights. Then, for $\bm a,\bm b\in\Hcal$,
\begin{equation}
	\inner{\bm a}{\bm b}_{\Hcal}
	\approx
	\sum_{i=1}^{M_{\rm T}}
	\sum_{j=1}^{M_{\rm T}}
	w_{x,i}^{\rm T}w_{y,j}^{\rm T}
	\bm a^H(\bm s_{ij})\bm b(\bm s_{ij}),
	\label{eq:quadrature_aperture}
\end{equation}
where $\bm s_{ij}=[s_{x,i},s_{y,j},0]^T$. Similarly, with target-region quadrature nodes and weights $\{\xi_{x,a}^{(q)},w_{x,a}^{(q)}\}$ and $\{\xi_{y,b}^{(q)},w_{y,b}^{(q)}\}$,
\begin{equation}
	\norm{\Acal_q\bm a}_{\Ycal_q}^2
	\approx
	\sum_{a=1}^{M_{\rm S}}
	\sum_{b=1}^{M_{\rm S}}
	w_{x,a}^{(q)}w_{y,b}^{(q)}
	\norm{(\Acal_q\bm a)(\bm\xi_{q,ab})}_2^2.
	\label{eq:quadrature_target}
\end{equation}
We set $M_{\rm T}=13$ and $M_{\rm S}=11$.

The adaptive enrichment is initialized with $\Ucal_c=\spann\{\bm h_1,\ldots,\bm h_K\}$ and terminates when $\Delta_m\le\epsilon_{\rm sp}$, with $\epsilon_{\rm sp}=10^{-7}$.
The limiting objective value is denoted by $p_\infty$ and its numerically converged value is used for normalization. For effective-rank evaluation, eigenvalues exceeding $10^{-3}$ of the largest eigenvalue are counted.

\subsection{Benchmark Schemes}
\label{subsec:benchmarks}

Two non-adaptive modal baselines are considered. All methods are initialized with the same communication response space $\Ucal_c$ and solve the same restricted SDR; they differ only in how the additional modes are selected.

\emph{1) Truncated-SVD:}
The additional modes are chosen as the dominant eigenmodes of the sensing operator $\Rcal$,
\begin{equation}
	\Rcal\bm v_\ell
	=
	\lambda_\ell\bm v_\ell,
	\qquad
	\lambda_1\geq\lambda_2\geq\cdots,
\end{equation}
and are appended to $\Ucal_c$ in descending order of $\lambda_\ell$.

\emph{2) 2-D Spatial Fourier:}
The aperture current is expanded over two-dimensional spatial harmonics
\begin{equation}\label{eq:last_main}
	\phi_{n_x,n_y}(\bm s)
	=
	\exp\!\left[
	j2\pi
	\left(
	\frac{n_xs_x}{L_x}
	+
	\frac{n_ys_y}{L_y}
	\right)
	\right],
\end{equation}
with $n_x,n_y\in\{-5,\ldots,5\}$. The resulting modes are introduced from low to high spatial frequency and orthogonalized with respect to the current search space.

In contrast, the proposed method selects each enrichment direction from the maximally violated full-space dual-slack block, thereby accounting for both sensing and communication requirements.

\subsection{Finite-Rank Structure and Convergence}
We first compare the eigenvalue spectra of the sensing operator $\Rcal$ and the optimal total covariance $\Qcal^\star$ in Fig.~\ref{fig:eigenvalue}. The spectrum of $\Rcal$ decays slowly, with many non-negligible eigenvalues over the displayed range, reflecting the rich electromagnetic response of the extended targets. In contrast, the spectrum of $\Qcal^\star$ drops sharply after only a few dominant components. Thus, although the sensing operator contains many effective spatial modes, the optimal transmission concentrates power on a low-dimensional set of electromagnetic directions. This behavior is consistent with the finite-rank covariance structure established in
Theorem~\ref{thm:finite_rank_sdr}.
\begin{figure}[H]
	\centering
	\includegraphics[width=0.8\linewidth,trim={0.6cm 0.05cm 0.6cm 0.6cm}, clip]{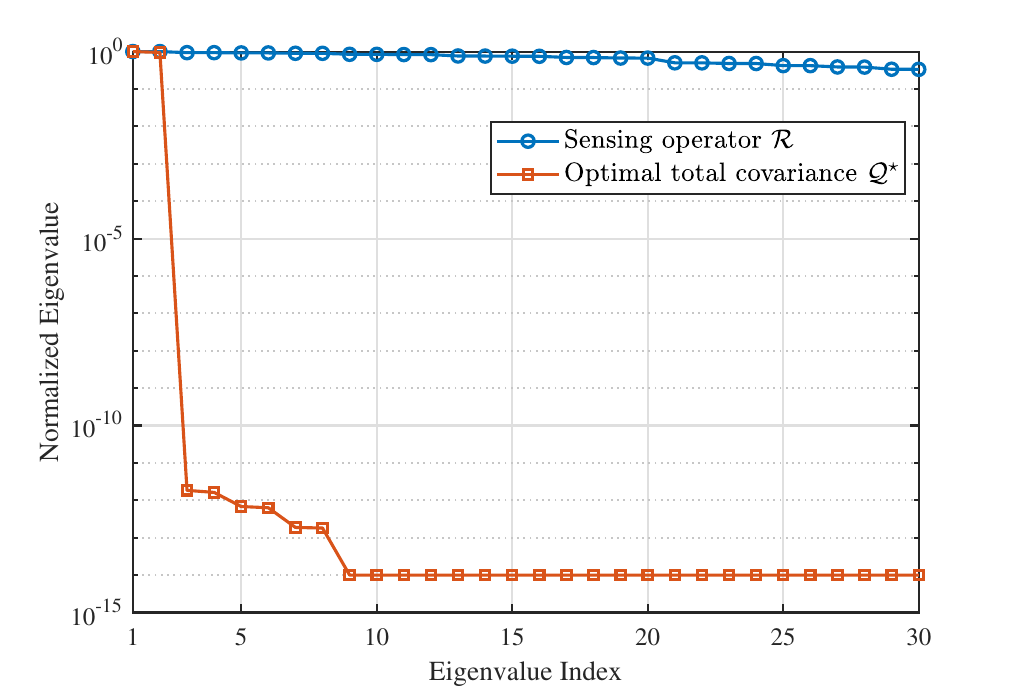}
	\caption{Normalized eigenvalue spectra of $\Rcal$ and $\Qcal^\star$.} 
	\label{fig:eigenvalue}
	\vspace{-4mm}
\end{figure}

To examine the effect of target extent, each target region is uniformly scaled about its center by a factor $\eta$. As shown in Fig.~\ref{fig:rank}, the effective rank of $\Rcal$ increases from about $50$ to more than $130$ as $\eta$ increases from $0.6$ to $1.4$, indicating that larger targets introduce more effective sensing modes. However, the effective rank of $\Qcal^\star$ remains nearly unchanged and substantially smaller than that of $\Rcal$. Thus, enlarging the target regions significantly enriches the sensing response, whereas the optimal transmission continues to rely on only a limited number of effective electromagnetic directions.
\begin{figure}[t]
	\centering
	\includegraphics[width=0.8\linewidth,trim={0.6cm 0.05cm 0.6cm 0.6cm}, clip]{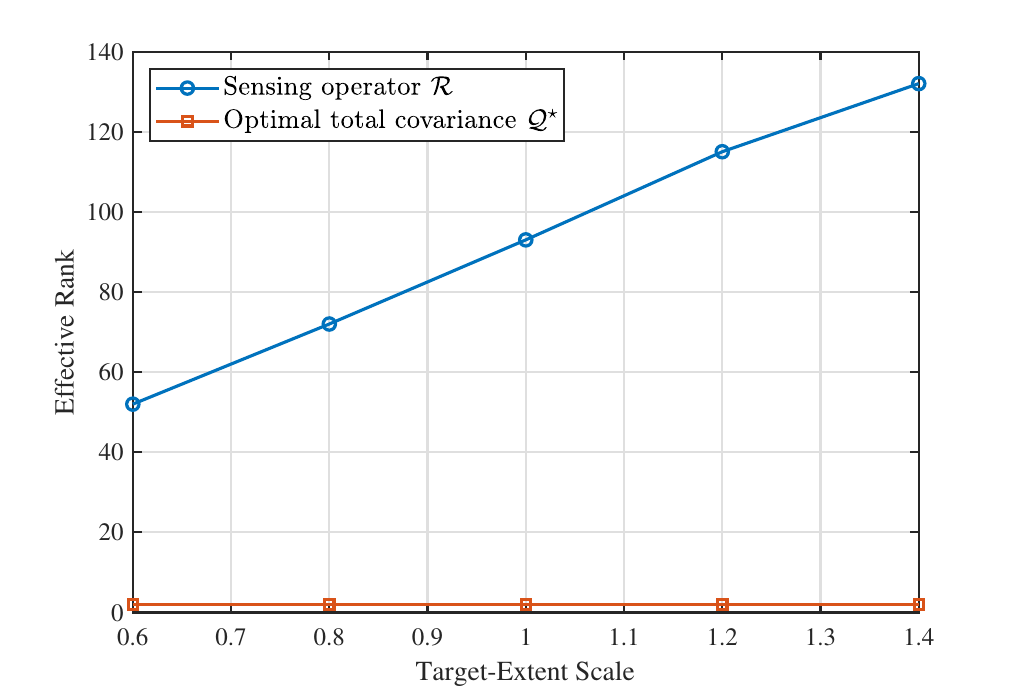}
	\caption{Effective ranks of $\Rcal$ and $\Qcal^\star$ versus target extent.} 
	\label{fig:rank}
	\vspace{-4mm}
\end{figure}

Fig.~\ref{fig:convergence} shows the convergence of the proposed adaptive enrichment, where $p_m$ is normalized by the limiting objective value $p_\infty$, and $m$ denotes the number of modes added to the initial communication response space. The normalized objective increases monotonically with the enriched search space. The initial communication subspace achieves about $83\%$ of $p_\infty$. After one and two added modes, the objective increases to about $90\%$ and above $95\%$, respectively, and approaches $p_\infty$ with only three added modes. Further enrichment provides only marginal improvement, showing that a small number of adaptively selected directions captures most of the attainable sensing performance.

\begin{figure}[t]
	\centering
	\includegraphics[width=0.8\linewidth,trim={0.6cm 0.05cm 0.6cm 0.6cm}, clip]{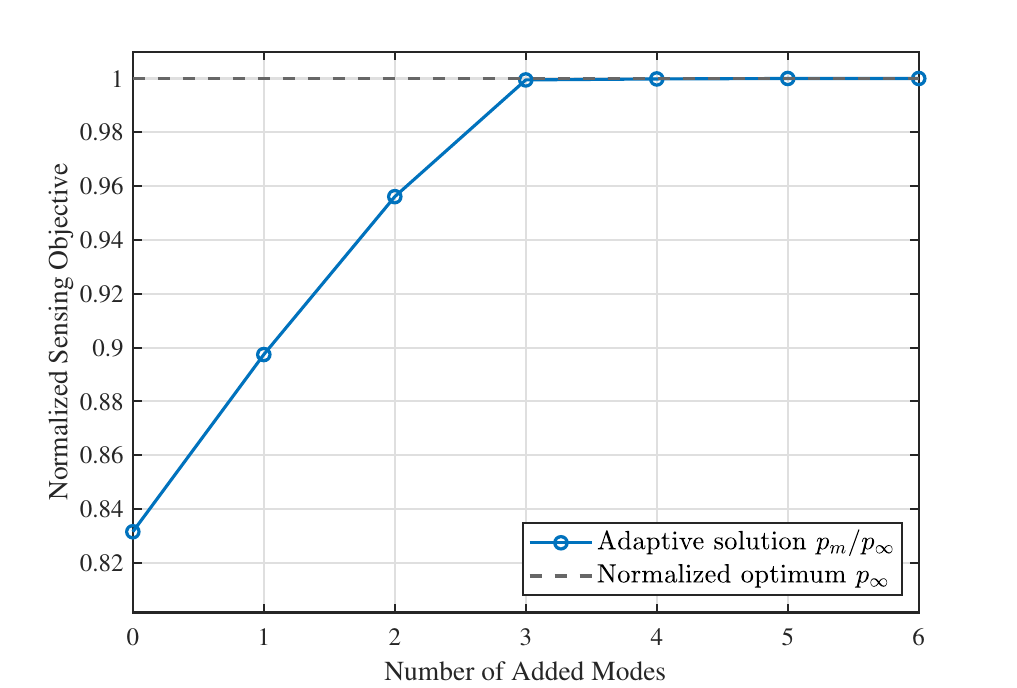}
	\caption{Convergence of the proposed adaptive enrichment. } 
	\label{fig:convergence}
	\vspace{-4mm}
\end{figure}

Fig.~\ref{fig:gap} compares the relative objective gap $(p_\infty-p_m)/p_\infty$ with the normalized certificate $[\Delta_m]_+P_{\max}/p_\infty$. Both decrease rapidly as the search space is enriched, while the certified bound remains above the observed gap throughout the iterations. This confirms that $\Delta_m$ provides a reliable measure of the residual full-space optimality gap. At the same time, the same quantity identifies whether additional electromagnetic modes are still needed and provides a natural stopping criterion for the adaptive enrichment.
\begin{figure}[t]
	\centering
	\includegraphics[width=0.8\linewidth,trim={0.6cm 0.05cm 0.6cm 0.6cm}, clip]{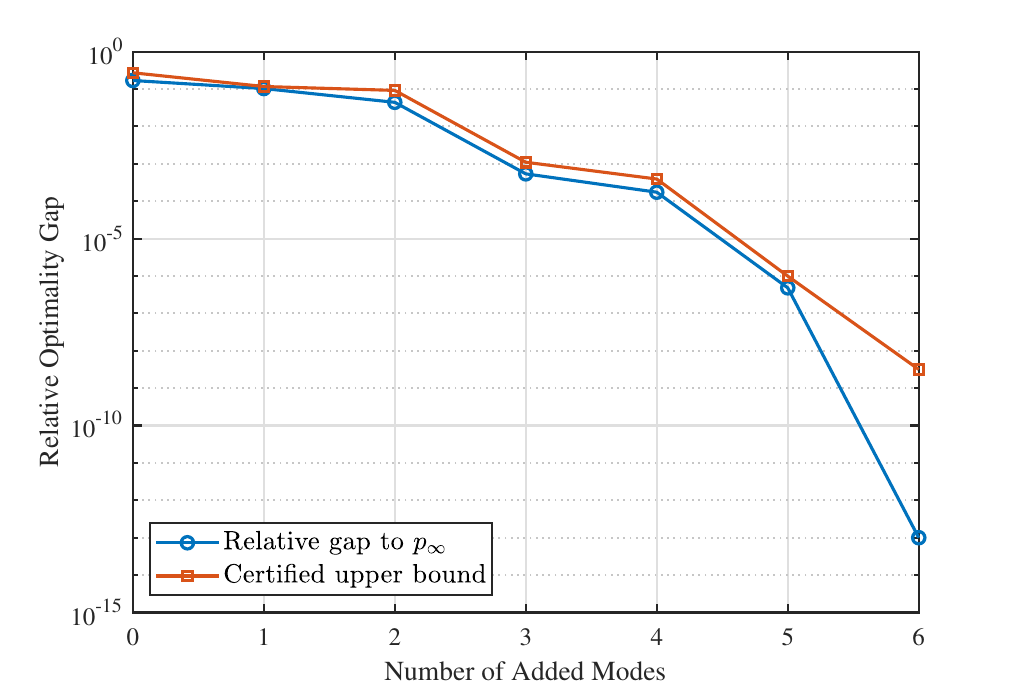}
	\caption{Relative objective gap and certified upper bound.}  
	\label{fig:gap}
	\vspace{-4mm}
\end{figure}

\subsection{Modal Efficiency and Sensing Performance}
\label{subsec:modal_results}
We next compare the proposed adaptive enrichment with the truncated-SVD and spatial-Fourier bases. Fig.~\ref{fig:effmode} shows the normalized total illumination versus the number of added modes. Since all methods start from the same communication response space, they have the same objective value when no additional mode is included. The proposed method improves rapidly and reaches nearly $p_{\infty}$ with only three added modes. The truncated-SVD basis also approaches the optimum as more sensing eigenmodes are included, but requires substantially more modes to achieve a comparable level. In contrast, the spatial-Fourier basis improves much more slowly and remains clearly below the optimal even after $30$ modes are added.
This difference stems from the mode-selection criterion. Truncated-SVD selects modes solely according to the sensing operator $\Rcal$, whereas the spatial-Fourier basis is determined only by the aperture geometry. In contrast, the proposed method selects each new mode according to the dual-slack violation of the full ISAC problem, so the added directions reflect both the sensing objective and the active communication constraints. As a result, near-optimal illumination can be achieved with far fewer modes.

\begin{figure}[H]
	\centering
	\includegraphics[width=0.8\linewidth,trim={0.6cm 0.05cm 0.6cm 0.6cm}, clip]{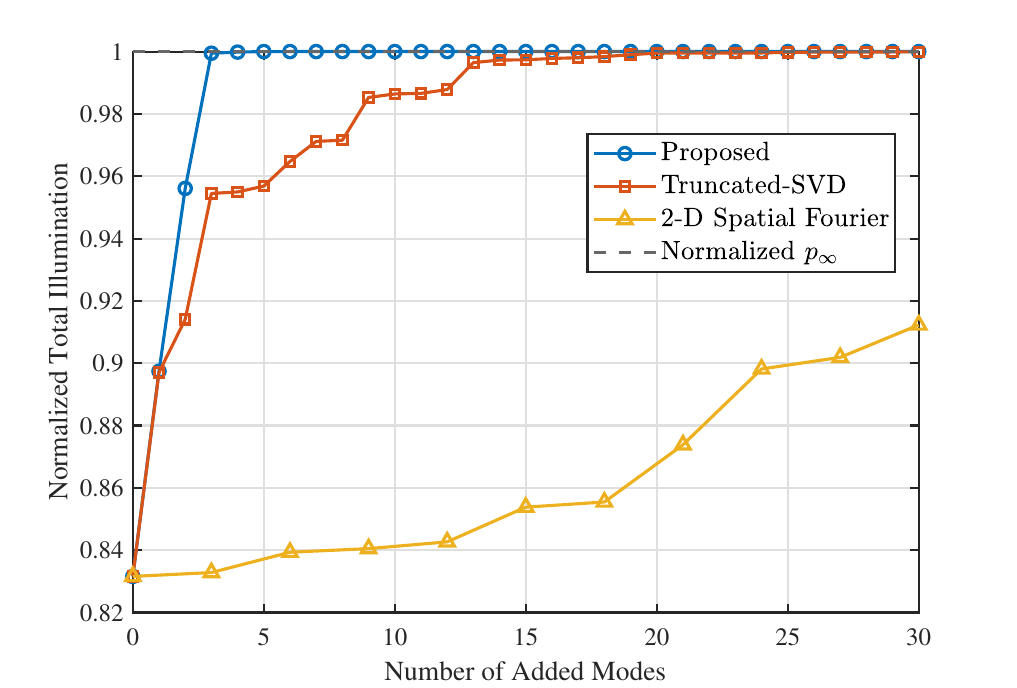}
	\caption{Modal efficiency of the proposed adaptive enrichment.} 
	\label{fig:effmode}
	\vspace{-4mm}
\end{figure}

The impact of the communication requirements is further examined in Fig.~\ref{fig:sinr}. The number of added modes is fixed to three, while the common SINR threshold varies from $8$ to $16$~dB. All illumination values are normalized by the high-accuracy reference at $8$~dB, so that the sensing degradation caused by stricter communication requirements remains visible. As the SINR threshold increases, the optimal value decreases, and the proposed method closely tracks it over the entire range. In contrast, the gap of truncated-SVD becomes more pronounced at high SINR, while the spatial-Fourier basis consistently yields lower illumination. The growing gap of truncated-SVD indicates that sensing-dominant modes are not always the most effective once the communication constraints become more restrictive.  By incorporating the dual information of the joint ISAC problem into mode selection, the proposed method better adapts to increasingly stringent communication constraints while maintaining high sensing performance.
\begin{figure}[t]
	\centering
	\includegraphics[width=0.8\linewidth,trim={0.6cm 0.05cm 0.6cm 0.6cm}, clip]{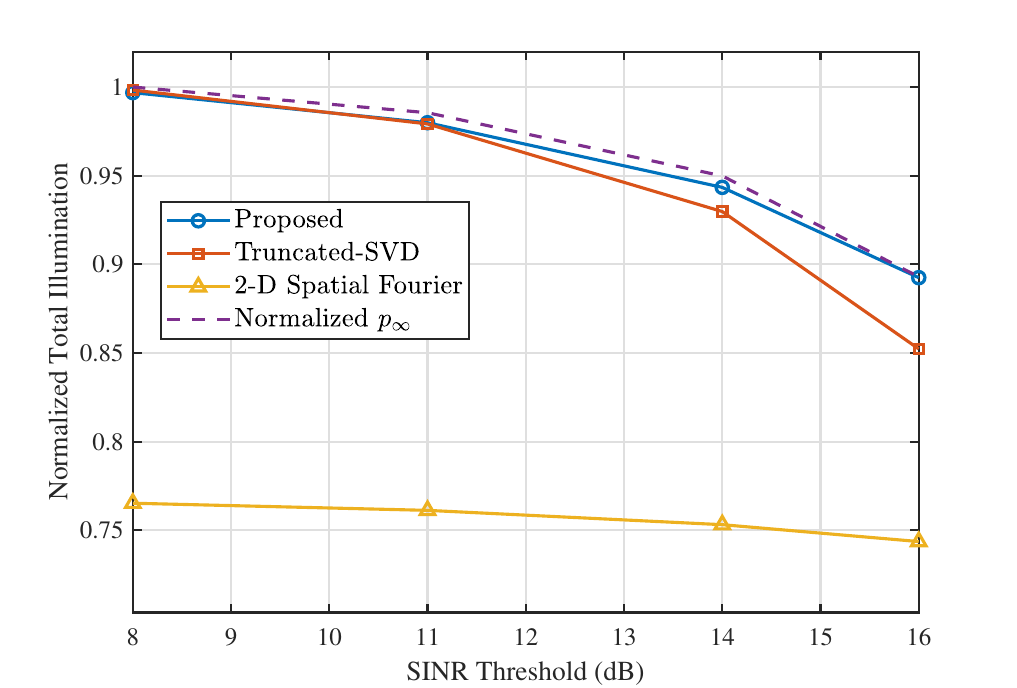}
	\caption{Communication--sensing tradeoff with three additional modes.} 
	\label{fig:sinr}
	\vspace{-4mm}
\end{figure}

\begin{figure*}[t]
	\centering
	\includegraphics[width=\linewidth,
	trim={0.6cm 1.2cm 0.2cm 0.8cm},
	clip]{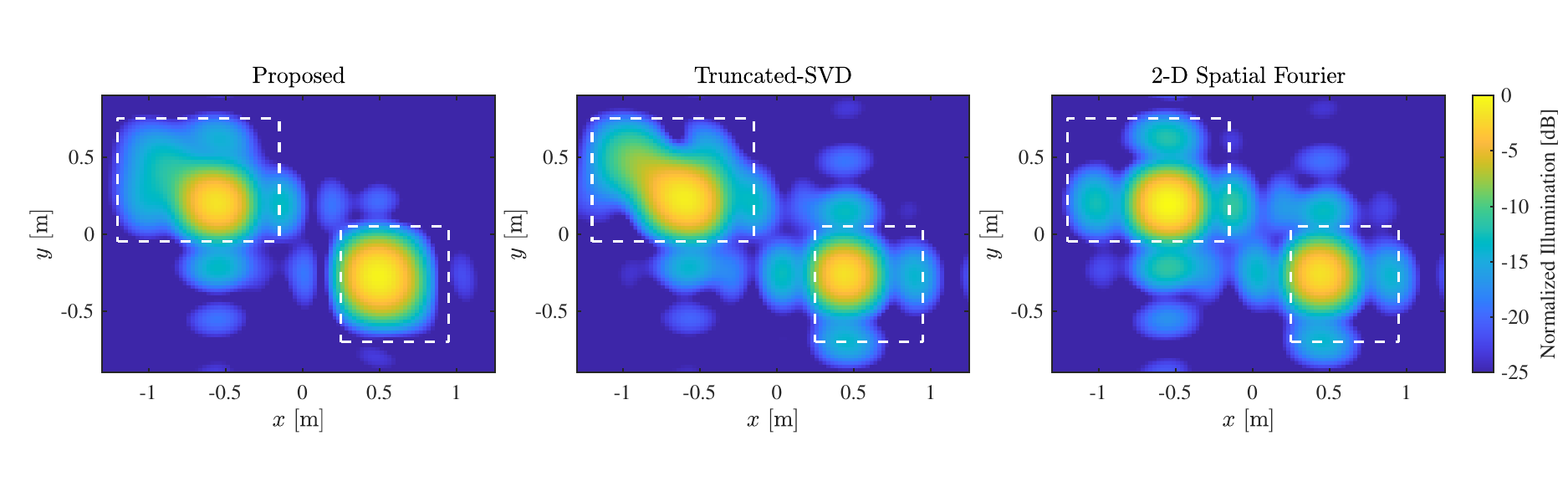}
	\caption{Extended-target illumination patterns at $z=2.3$~m under a
		common two-mode budget. The dashed rectangles indicate the
		prescribed target supports.}
	\label{fig:illu}
	\vspace{-4mm}
\end{figure*}

Fig.~\ref{fig:illu} compares the illumination patterns of the three methods under the same two-mode budget. All panels use the same normalization, and the dashed rectangles indicate the target supports. The proposed method concentrates the illumination mainly within the two target regions, with relatively limited leakage outside the prescribed
supports. The truncated-SVD design also illuminates both targets effectively, but exhibits more noticeable energy spread into the surrounding area. The spatial-Fourier design provides the least concentrated illumination, with stronger radiation outside the target regions. These results show that the proposed enrichment strategy better aligns the synthesized sensing field with the prescribed target regions.

\section{Conclusion}
\label{sec:conclusion}

This paper studied continuous-domain CAPA-ISAC beamforming for spatially extended sensing targets.
We showed that the resulting infinite-rank target response does not
necessarily require an infinite-dimensional transmit solution; instead,
the optimal transmission remains supported by a finite-dimensional
electromagnetic subspace. This shifts the main difficulty from optimizing
over the full function space to identifying the unknown modal support of
the optimum. Motivated by this structure, we developed an adaptive mode-enrichment
framework that discovers the required electromagnetic modes directly from
full-space optimality information, avoiding a prescribed basis or
truncation. The resulting sequence of finite-dimensional problems converges
to the full-space optimum and achieves high modal efficiency in the
numerical experiments. These results suggest adaptive subspace discovery as
a promising principle for broader continuous-aperture optimization
problems.

\appendices
\section{Proof of Lemma~\ref{lem:sdr_existence}}
\label{app:sdr_existence}

Let
\begin{equation}
	\mathfrak X
	=
	\mathcal S_1(\Hcal)^{K+1}
\end{equation}
be endowed with the componentwise weak-* topology \(\sigma(\mathcal S_1(\Hcal),\mathcal K(\Hcal))\). Since \(\mathcal S_1(\Hcal)=\mathcal K(\Hcal)^*\), the Banach--Alaoglu theorem implies that every trace-norm bounded subset of \(\mathfrak X\) is relatively compact in this topology. For any feasible covariance tuple
\((\{\Qcal_k\}_{k=1}^{K},\Qcal_s)\), positivity gives
\begin{equation}
	\norm{\Qcal_k}_1=\Tr(\Qcal_k),
	\qquad
	\norm{\Qcal_s}_1=\Tr(\Qcal_s),
\end{equation}
and the power constraint yields
\begin{equation}
	\sum_{k=1}^{K}\norm{\Qcal_k}_1
	+\norm{\Qcal_s}_1
	\le P_{\max}.
	\label{eq:app_trace_bound}
\end{equation}
Hence the feasible set is contained in a weak-* compact product ball.

It remains to show that the feasible set is weak-* closed. Consider a sequence of feasible covariance tuples converging componentwise weak-* to \((\{\Qcal_k\}_{k=1}^{K},\Qcal_s)\). For any \(\bm u\in\Hcal\), the rank-one operator \(\bm u\otimes\bm u\) is compact. Thus, for every component \(\Qcal_{i,n}\rightharpoonup^*\Qcal_i\),
\begin{equation}
	\inner{\bm u}{\Qcal_i\bm u}_{\Hcal}
	=
	\lim_n
	\inner{\bm u}{\Qcal_{i,n}\bm u}_{\Hcal}
	\ge0,
\end{equation}
which proves \(\Qcal_i\succeq0\); the same holds for \(\Qcal_s\). Moreover, the dual norm on \(\mathcal S_1(\Hcal)\) is weak-* lower semicontinuous. Using positivity,
\begin{align*}
	\sum_{k=1}^{K}\Tr(\Qcal_k)+\Tr(\Qcal_s)
	&=
	\sum_{k=1}^{K}\norm{\Qcal_k}_1+\norm{\Qcal_s}_1
	\\
	&\le
	\liminf_n
	(
	\sum_{k=1}^{K}\Tr(\Qcal_{k,n})
	+\Tr(\Qcal_{s,n})
	)
	\\
	&\le P_{\max},
\end{align*}
so the power constraint is preserved.

For the SINR constraints, each \(\Hcal_k=\bm h_k\otimes\bm h_k\) is rank one and hence compact. Thus \(\Qcal\mapsto\Tr(\Hcal_k\Qcal)\) is weak-* continuous, so the affine SINR inequalities in \eqref{prob:sdr_sinr} define weak-* closed sets. Therefore the feasible set of \eqref{prob:sdr} is weak-* compact.
Finally, \(\Rcal\) is compact since each target illumination operator is induced by a Hilbert--Schmidt aperture-to-region operator, and \(\Rcal\) is their finite weighted sum. Hence \(\Qcal\mapsto\Tr(\Rcal\Qcal)\) is weak-* continuous. Since the feasible set is nonempty and weak-* compact, the objective attains its maximum ~\cite{kreyszig1978introductory} at a feasible point. Thus \eqref{prob:sdr} admits an optimal solution.

\vspace{-5pt}
\section{Proof of Theorem~\ref{thm:finite_rank_response}}
\label{app:finite_rank_response}

Let $\Pcal_{\rm fr}$ be the orthogonal projector onto
$\Ucal_{\rm fr}$ in \eqref{eq:finite_response_space}. For any feasible
tuple $(\{\Qcal_k\}_{k=1}^{K},\Qcal_s)$ of \eqref{prob:sdr}, define
\begin{equation}
	\overline{\Qcal}_j
	=
	\Pcal_{\rm fr}\Qcal_j\Pcal_{\rm fr},
	\qquad
	j\in\{s,1,\ldots,K\}.
	\label{eq:app_projected_covariances}
\end{equation}
Each $\overline{\Qcal}_j$ is positive semidefinite and trace class, with
\begin{equation}
	\ran(\overline{\Qcal}_j)
	\subseteq
	\Ucal_{\rm fr}.
	\label{eq:app_projected_range}
\end{equation}
Since $\bm h_k\in\Ucal_{\rm fr}$,
$\Pcal_{\rm fr}\bm h_k=\bm h_k$, and hence
\begin{align}
	\Tr(\Hcal_k\overline{\Qcal}_j)
	&=
	\inner{\bm h_k}{\Pcal_{\rm fr}\Qcal_j\Pcal_{\rm fr}\bm h_k}_{\Hcal}
	\nonumber\\
	&=
	\inner{\bm h_k}{\Qcal_j\bm h_k}_{\Hcal}
	=
	\Tr(\Hcal_k\Qcal_j).
	\label{eq:app_projected_channel_terms}
\end{align}
so all SINR constraints are preserved.

Since $\Rcal$ is self-adjoint and $\ran(\Rcal)\subseteq\Ucal_{\rm fr}$, its orthogonal complement is contained in $\ker(\Rcal)$, which gives
$\Rcal=\Pcal_{\rm fr}\Rcal\Pcal_{\rm fr}$. By trace cyclicity,
\begin{align}
	\Tr(\Rcal\overline{\Qcal}_j)
	=
	\Tr(\Rcal\Pcal_{\rm fr}\Qcal_j\Pcal_{\rm fr})
	=
	\Tr(\Rcal\Qcal_j),
	\label{eq:app_projected_sensing_terms}
\end{align}
and the sensing objective is unchanged. Moreover, since $0\preceq\Pcal_{\rm fr}\preceq\Ical$,
\begin{align}
	\Tr(\overline{\Qcal}_j)
	=
	\Tr(\Qcal_j^{1/2}\Pcal_{\rm fr}\Qcal_j^{1/2})\le
	\Tr(\Qcal_j),
	\label{eq:app_projected_trace}
\end{align}
so the projection does not increase the transmit power. Hence, every feasible point of \eqref{prob:sdr} can be projected onto $\Ucal_{\rm fr}$ while preserving the SINR values and sensing objective.

Applying this construction to an optimal solution, whose existence follows
from Lemma~\ref{lem:sdr_existence}, yields another optimal solution satisfying
\eqref{eq:finite_response_range}. Since $\Ucal_{\rm fr}$ is finite dimensional,
\begin{equation}
	\rank(\overline{\Qcal}_j)
	=
	\dim\ran(\overline{\Qcal}_j)
	\le
	\dim(\Ucal_{\rm fr})
	<\infty,
\end{equation}
which proves \eqref{eq:finite_response_rank_bound} and completes the proof.

\section{Proof of Theorem~\ref{thm:finite_rank_sdr}}
\label{app:finite_rank_theorem}

Let $\{\Qcal_k^\star\}_{k=1}^{K}$ and $\Qcal_s^\star$ be an optimal primal solution, with optimal dual multipliers $\{\lambda_k^\star\}_{k=1}^{K}$ and $\mu^\star$ satisfying the full-space KKT conditions. Define
\begin{align}
	\Bcal_s^\star
	&=
	\Rcal-\sum_{k=1}^{K}\gamma_k\lambda_k^\star\Hcal_k,
	\label{eq:app_Bs_star}
	\\
	\Bcal_i^\star
	&=
	\Rcal+\lambda_i^\star\Hcal_i
	-\sum_{k\neq i}\gamma_k\lambda_k^\star\Hcal_k,
	\quad i=1,\ldots,K,
	\label{eq:app_Bi_star}
\end{align}
and the dual-slack operators
\begin{equation}
	\Zcal_j^\star
	=
	\mu^\star\Ical-\Bcal_j^\star
	\succeq0,
	\qquad j\in\{s,1,\ldots,K\}.
	\label{eq:app_Zstar_finite_rank}
\end{equation}
The full-space complementary-slackness conditions give
\begin{equation}
	\Tr(\Zcal_j^\star\Qcal_j^\star)=0,
	\qquad j\in\{s,1,\ldots,K\}.
	\label{eq:app_full_complementarity}
\end{equation}
Since both operators are positive semidefinite,
\begin{align}
	0
	&=
	\Tr\!\left(
	(\Qcal_j^\star)^{1/2}
	\Zcal_j^\star
	(\Qcal_j^\star)^{1/2}
	\right)
	\nonumber\\
	&=
	\left\|
	(\Zcal_j^\star)^{1/2}
	(\Qcal_j^\star)^{1/2}
	\right\|_{\rm HS}^{2}.
	\label{eq:app_HS_zero}
\end{align}
Hence $(\Zcal_j^\star)^{1/2}(\Qcal_j^\star)^{1/2}=0$, and therefore
\begin{equation}
	\operatorname{supp}(\Qcal_j^\star)
	\subseteq
	\ker(\Zcal_j^\star).
	\label{eq:app_support_kernel}
\end{equation}
Since $\Ucal_c=\spann\{\bm h_1,\ldots,\bm h_K\}$ is finite dimensional while $\Rcal$ has infinite rank, there exists a unit vector $\bm u\in\Ucal_c^\perp$ such that
\begin{equation}
	\inner{\bm u}{\Rcal\bm u}_{\Hcal}>0.
	\label{eq:app_positive_sensing_direction}
\end{equation}
Indeed, if $\inner{\bm u}{\Rcal\bm u}_{\Hcal}=0$ for all $\bm u\in\Ucal_c^\perp$, then the positivity of $\Rcal$ would imply $\ran(\Rcal^{1/2})\subseteq\Ucal_c$, forcing $\Rcal$ to have finite rank. Since $\Hcal_k\bm u=0$ for all $k$, dual feasibility of $\Zcal_s^\star$ yields
\begin{equation}
	0
	\le
	\inner{\bm u}{\Zcal_s^\star\bm u}_{\Hcal}
	=
	\mu^\star
	-
	\inner{\bm u}{\Rcal\bm u}_{\Hcal},
\end{equation}
so $\mu^\star>0$.

Finally, $\Rcal$ is compact and every $\Hcal_k$ is rank one, so each $\Bcal_j^\star$ is compact and self-adjoint. From \eqref{eq:app_Zstar_finite_rank} and $\mu^\star>0$, 
\begin{equation}
	\ker(\Zcal_j^\star)
	=
	\ker(\mu^\star\Ical-\Bcal_j^\star)
	\label{eq:app_dual_kernel_eigenspace}
\end{equation}
is either trivial or the eigenspace of $\Bcal_j^\star$ associated with the nonzero eigenvalue $\mu^\star$. Since every nonzero eigenvalue of a compact operator has finite multiplicity~[Theorem~8.3-1,~\citen{kreyszig1978introductory}], $\dim\ker(\Zcal_j^\star)<\infty$. Together with \eqref{eq:app_support_kernel}, this gives
\begin{equation}
	\rank(\Qcal_j^\star)<\infty,
	\qquad j\in\{s,1,\ldots,K\},
\end{equation}
which proves Theorem~\ref{thm:finite_rank_sdr}.

\section{Proof of Theorem~\ref{thm:sdr_tight}}
\label{app:sdr_tightness}

From \eqref{prob:sdr_sinr} and \(\Qcal-\Qcal_k\succeq0\), we have
\begin{equation}
	\Tr(\Hcal_k\Qcal_k^\star)
	\ge
	\gamma_k\sigma_k^2
	>
	0,
	\qquad
	k=1,\ldots,K.
	\label{eq:app_positive_signal}
\end{equation}
Therefore, $\widehat{\bm w}_k$ in
\eqref{eq:rank_recovery} is well defined and nonzero, so that
\begin{equation}
	\rank(\widehat{\Qcal}_k)=1,
	\qquad
	k=1,\ldots,K.
	\label{eq:app_rank_one}
\end{equation}

We next show that
\begin{equation}
	\widehat{\Qcal}_k
	\preceq
	\Qcal_k^\star.
	\label{eq:app_Qhat_le_Q}
\end{equation}
Let
$
\bm a_k
\triangleq
\frac{
	(\Qcal_k^\star)^{1/2}\bm h_k
}{
	\sqrt{\Tr(\Hcal_k\Qcal_k^\star)}
}.
$
Since
\begin{equation}
	\Tr(\Hcal_k\Qcal_k^\star)
	=
	\inner{\bm h_k}
	{\Qcal_k^\star\bm h_k}_{\Hcal}
	=
	\norm{
		(\Qcal_k^\star)^{1/2}\bm h_k
	}_{\Hcal}^{2},
\end{equation}
we have $\norm{\bm a_k}_{\Hcal}=1$. Moreover,
\begin{align}
	\widehat{\Qcal}_k
	=
	\widehat{\bm w}_k
	\otimes
	\widehat{\bm w}_k
	=
	(\Qcal_k^\star)^{1/2}
	(\bm a_k\otimes\bm a_k)
	(\Qcal_k^\star)^{1/2}.
	\label{eq:app_Qhat_factor}
\end{align}
Therefore,
\begin{align}
	\Qcal_k^\star-\widehat{\Qcal}_k
	&=
	(\Qcal_k^\star)^{1/2}
	\left(
	\Ical-\bm a_k\otimes\bm a_k
	\right)
	(\Qcal_k^\star)^{1/2}
	\succeq0,
	\label{eq:app_residual_psd}
\end{align}
which proves~\eqref{eq:app_Qhat_le_Q}.

Since $\Qcal_k^\star$ and $\widehat{\Qcal}_k$ are trace class,
their difference is also trace class. It follows from
\eqref{eq:Qs_recovery} and~\eqref{eq:app_residual_psd} that
\begin{equation}
	\widehat{\Qcal}_s
	=
	\Qcal_s^\star
	+
	\sum_{k=1}^{K}
	\left(
	\Qcal_k^\star-\widehat{\Qcal}_k
	\right)
	\in
	\mathcal S_1^+(\Hcal).
	\label{eq:app_Qs_psd}
\end{equation}

We now verify that the recovery preserves the desired-signal terms and the
aggregate transmit covariance. From~\eqref{eq:rank_recovery},
\begin{align}
	\left|
	\inner{\bm h_k}
	{\widehat{\bm w}_k}_{\Hcal}
	\right|^2
	&=
	\frac{
		\left|
		\inner{\bm h_k}
		{\Qcal_k^\star\bm h_k}_{\Hcal}
		\right|^2
	}{
		\Tr(\Hcal_k\Qcal_k^\star)
	}
	\nonumber\\
	&=
	\Tr(\Hcal_k\Qcal_k^\star),
	\qquad
	k=1,\ldots,K,
	\label{eq:app_desired_preserved}
\end{align}
where we used $\inner{\bm h_k} {\Qcal_k^\star\bm h_k}_{\Hcal}=\Tr(\Hcal_k\Qcal_k^\star)$. Equivalently,
\begin{equation}
	\Tr(\Hcal_k\widehat{\Qcal}_k)
	=
	\Tr(\Hcal_k\Qcal_k^\star).
	\label{eq:app_signal_trace_preserved}
\end{equation}
Furthermore, by~\eqref{eq:Qs_recovery},
\begin{align}
	\sum_{k=1}^{K}\widehat{\Qcal}_k
	+
	\widehat{\Qcal}_s
	&=
	\sum_{k=1}^{K}\widehat{\Qcal}_k
	+
	\Qcal_s^\star
	+
	\sum_{k=1}^{K}
	\left(
	\Qcal_k^\star-\widehat{\Qcal}_k
	\right)
	\nonumber\\
	&=
	\sum_{k=1}^{K}\Qcal_k^\star
	+
	\Qcal_s^\star
	=
	\Qcal^\star.
	\label{eq:app_total_cov_preserved}
\end{align}
Equations~\eqref{eq:app_signal_trace_preserved} and \eqref{eq:app_total_cov_preserved} show that all communication SINRs, the transmit power, and the illumination objective are preserved exactly. Therefore, $\{\widehat{\bm w}_k\}_{k=1}^{K}$ together with $\widehat{\Qcal}_s$ is feasible for \eqref{prob:P0} and attains the optimal value of $\mathcal P_{\rm SDR}$. Since $\mathcal P_{\rm SDR}$ is a relaxation of $\mathcal P_0$, the recovered solution is globally optimal for \eqref{prob:P0}. This proves Theorem~\ref{thm:sdr_tight}.

\section{Proof of Fact~\ref{fact:restricted_full_kkt}}
\label{app:restricted_full_kkt}

Consider the primal--dual optimal KKT point of $\mathcal P_m$ introduced in Section~\ref{subsec:optimality_conditions}, with $\widehat{\Qcal}_{j,m}$ defined in
\eqref{eq:lifted_restricted_solution}. Since
\[
\bm H_{k,m}
=
\bm\Psi_m^*
\Hcal_k
\bm\Psi_m,
\qquad
\bm R_m
=
\bm\Psi_m^*
\Rcal
\bm\Psi_m,
\]
and $\bm\Psi_m^*\bm\Psi_m=\bm I_{r_m}$,
trace cyclicity gives
\begin{align}
	\Tr(
	\Hcal_k\widehat{\Qcal}_{j,m}
	)
	&=
	\Tr(
	\bm H_{k,m}\bm Q_{j,m}^\star
	),
	\label{eq:app_lift_channel}
	\\
	\Tr(
	\widehat{\Qcal}_{j,m}
	)
	&=
	\Tr(
	\bm Q_{j,m}^\star
	),
	\label{eq:app_lift_power}
	\\
	\Tr(
	\Rcal\widehat{\Qcal}_{j,m}
	)
	&=
	\Tr(
	\bm R_m\bm Q_{j,m}^\star
	).
	\label{eq:app_lift_sensing}
\end{align}
Moreover, $\widehat{\Qcal}_{j,m}\succeq0$ whenever $\bm Q_{j,m}^\star\succeq0$. Hence, \eqref{eq:restricted_KKT_primal_sinr} and \eqref{eq:restricted_KKT_primal_other} imply
\eqref{eq:full_KKT_primal_sinr} and \eqref{eq:full_KKT_primal_other}, respectively. The scalar dual-feasibility condition \eqref{eq:full_KKT_scalar_dual} follows directly from \eqref{eq:restricted_KKT_scalar_dual}. Since all scalar constraint residuals are preserved by the lifting, the scalar complementary-slackness conditions \eqref{eq:full_KKT_comp_sinr} and \eqref{eq:full_KKT_comp_power} follow from \eqref{eq:restricted_KKT_comp_sinr} and \eqref{eq:restricted_KKT_comp_power}.

From the definitions of the restricted and full-space pricing operators,
\begin{equation}
	\bm B_{j,m}^{(r)}
	=
	\bm\Psi_m^*\Bcal_{j,m}\bm\Psi_m,
	\label{eq:app_pricing_compression}
\end{equation}
and hence
\begin{equation}
	\bm\Psi_m^*\Zcal_{j,m}\bm\Psi_m
	=
	\mu_m\bm I_{r_m}-\bm B_{j,m}^{(r)}
	=
	\bm Z_{j,m}^{(r)}.
	\label{eq:app_slack_compression}
\end{equation}
Using~\eqref{eq:lifted_restricted_solution}, \eqref{eq:app_slack_compression}, and trace cyclicity, we obtain
\begin{align}
	\Tr(\Zcal_{j,m}\widehat{\Qcal}_{j,m})
	&=
	\Tr\!\left(
	\bm Z_{j,m}^{(r)}
	\bm Q_{j,m}^{\star}
	\right)
	=0,
	\label{eq:app_lifted_complementarity}
\end{align}
which establishes~\eqref{eq:full_KKT_comp_operator}.

The restricted KKT condition guarantees only
\begin{equation}
	\bm\Psi_m^*\Zcal_{j,m}\bm\Psi_m
	=
	\bm Z_{j,m}^{(r)}
	\succeq0,
	\label{eq:app_restricted_slack_psd}
\end{equation}
which ensures nonnegativity of the quadratic form over $\Ucal_m$, but not $\Zcal_{j,m}\succeq0$ over $\Hcal$. Thus, \eqref{eq:full_KKT_dual_slack} is the only full-space KKT condition not automatically inherited from the restricted KKT point, proving Fact~\ref{fact:restricted_full_kkt}.

\section{Proof of Proposition~\ref{prop:gap}}
\label{app:gap_bound}

By the lifting construction in~\eqref{eq:restricted_covariances}, every
feasible solution of the restricted problem~\eqref{prob:Pm} induces a
feasible solution of the full-space SDR~\eqref{prob:sdr} with the same
objective value. Hence,
\begin{equation}
	p_m\le p^\star.
	\label{eq:app_pm_le_pstar}
\end{equation}

Let $\{\lambda_{k,m}\}_{k=1}^{K}$ and $\mu_m$ be an optimal dual
solution of~\eqref{prob:Pm}. By the definitions of $\Delta_m$ and
$\delta_m=[\Delta_m]_+$,
\begin{equation}
	\rho_{j,m}
	\le
	\mu_m+\delta_m,
	\qquad
	j\in\{s,1,\ldots,K\}.
	\label{eq:rho_shift_bound}
\end{equation}
Since
$\rho_{j,m}$ is the largest Rayleigh quotient of
$\Bcal_{j,m}$, \eqref{eq:rho_shift_bound} implies
\begin{equation}
	(\mu_m+\delta_m)\Ical
	-
	\Bcal_{j,m}
	\succeq0,
	\qquad
	j\in\{s,1,\ldots,K\}.
	\label{eq:shifted_full_dual_feasible}
\end{equation}
Therefore, $\{\lambda_{k,m}\}_{k=1}^{K}$ together with
\begin{equation}
	\widetilde{\mu}_m
	=
	\mu_m+\delta_m
\end{equation}
constitutes a feasible dual point of the full-space
SDR~\eqref{prob:sdr}.

Let $d_m$ denote the optimal dual value of the restricted
SDP~\eqref{prob:Pm}. By strong duality,
\begin{equation}
	d_m=p_m.
\end{equation}
Increasing the power multiplier from $\mu_m$ to
$\widetilde{\mu}_m$ changes the dual objective by
$\delta_m P_{\max}$. Hence, the objective value of the above
full-space dual-feasible point is
\begin{equation}
	p_m+\delta_mP_{\max}.
\end{equation}
By weak duality for~\eqref{prob:sdr},
\begin{equation}
	p^\star
	\le
	p_m+\delta_mP_{\max}.
	\label{eq:app_gap_upper}
\end{equation}
Combining~\eqref{eq:app_pm_le_pstar} and
\eqref{eq:app_gap_upper} gives
\begin{equation}
	0
	\le
	p^\star-p_m
	\le
	\delta_mP_{\max},
\end{equation}
which proves~\eqref{eq:gap_bound}.

In particular, if $\Delta_m\le0$, then $\delta_m=0$, and hence
$p_m=p^\star$. Since the lifted restricted solution is feasible for
\eqref{prob:sdr} and attains $p_m$, it is globally optimal for the
full-space SDR.

\section{Proof of Theorem~\ref{thm:global_convergence}}
\label{app:global_convergence}

If the exact enrichment procedure terminates at some finite iteration $m$ with $\Delta_m\le0$, Proposition~\ref{prop:gap} gives $p_m=p^\star$. We therefore consider the case where the enrichment procedure generates an infinite sequence.

We first establish
\begin{equation}
	\delta_m\longrightarrow0.
	\label{eq:app_delta_to_zero}
\end{equation}
Suppose, by contradiction, that \eqref{eq:app_delta_to_zero} does not hold. Then there exist $\varepsilon>0$ and an infinite subsequence
$\{m_\ell\}$ such that
\begin{equation}
	\delta_{m_\ell}\ge\varepsilon,
	\qquad
	\forall\,\ell.
	\label{eq:delta_lower_subseq}
\end{equation}
Since $\delta_{m_\ell}>0$, we have $\Delta_{m_\ell}=\delta_{m_\ell}$. Because there are only $K+1$ covariance blocks, a further subsequence can be selected such that the
same index $j^\star\in\{s,1,\ldots,K\}$ attains the maximum dual-slack violation. For notational simplicity, we retain the notation $\{m_\ell\}$.

Let $\bm v_\ell\triangleq\bm v_{m_\ell}$ denote the unit-norm violating eigenfunction identified at iteration $m_\ell$. By the enrichment rule,
\begin{equation}
	\Bcal_{j^\star,m_\ell}\bm v_\ell
	=
	\rho_{j^\star,m_\ell}\bm v_\ell,
	\qquad
	\norm{\bm v_\ell}_{\Hcal}=1,
	\label{eq:conv_eig_relation}
\end{equation}
where
\begin{equation}
	\rho_{j^\star,m_\ell}
	=
	\mu_{m_\ell}
	+
	\delta_{m_\ell}
	\ge
	\varepsilon.
	\label{eq:conv_positive_eig}
\end{equation}

By Assumption~\ref{ass:dual_bounded}, the nonnegative multiplier sequence
\(
\left\{
\left(
\lambda_{1,m},
\ldots,
\lambda_{K,m},
\mu_m
\right)
\right\}
\)
is uniformly bounded. Hence, after passing to a further subsequence if necessary,
\begin{equation}
	\lambda_{k,m_\ell}
	\longrightarrow
	\lambda_{k,\infty},
	\qquad
	\mu_{m_\ell}
	\longrightarrow
	\mu_\infty.
	\label{eq:dual_subseq_conv}
\end{equation}
Since each pricing operator depends affinely on the scalar dual
multipliers, \eqref{eq:dual_subseq_conv} implies
\begin{equation}
	\norm{
		\Bcal_{j^\star,m_\ell}
		-
		\Bcal_{j^\star,\infty}
	}_{\rm op}
	\longrightarrow0,
	\label{eq:B_op_convergence}
\end{equation}
where $\Bcal_{j^\star,\infty}$ is obtained by replacing the dual
multipliers with their limits. Since $\Rcal$ and $\Hcal_k$ are compact
and self-adjoint, so is $\Bcal_{j^\star,\infty}$.

Since $\{\Bcal_{j^\star,m_\ell}\}$ is uniformly bounded in operator norm,
\begin{align}
	\rho_{j^\star,m_\ell}
	=
	\inner{\bm v_\ell}
	{\Bcal_{j^\star,m_\ell}\bm v_\ell}_{\Hcal}
	\le
	\norm{\Bcal_{j^\star,m_\ell}}_{\rm op}
	\norm{\bm v_\ell}_{\Hcal}^{2},
\end{align}
and hence $\{\rho_{j^\star,m_\ell}\}$ is bounded. Together with
\eqref{eq:conv_positive_eig}, this allows us to extract a further subsequence such that
\begin{equation}
	\rho_{j^\star,m_\ell}
	\longrightarrow
	\rho_\infty
	\ge
	\varepsilon.
	\label{eq:rho_subseq_conv}
\end{equation}

Since $\Bcal_{j^\star,\infty}$ is compact and
$\{\bm v_\ell\}$ is bounded, there exists a further subsequence, still
indexed by $\ell$, such that
\begin{equation}
	\Bcal_{j^\star,\infty}\bm v_\ell
	\longrightarrow
	\bm w_\infty
	\quad
	\text{strongly in }\Hcal
	\label{eq:Binf_v_conv}
\end{equation}
for some $\bm w_\infty\in\Hcal$. Moreover, by
\eqref{eq:B_op_convergence} and
$\norm{\bm v_\ell}_{\Hcal}=1$,
\begin{align}
	&
	\norm{
		\Bcal_{j^\star,m_\ell}\bm v_\ell
		-
		\Bcal_{j^\star,\infty}\bm v_\ell
	}_{\Hcal}
	\nonumber\\
	&\qquad\le
	\norm{
		\Bcal_{j^\star,m_\ell}
		-
		\Bcal_{j^\star,\infty}
	}_{\rm op}
	\longrightarrow0.
	\label{eq:Bv_difference}
\end{align}
Combining~\eqref{eq:Binf_v_conv} and \eqref{eq:Bv_difference} yields
\begin{equation}
	\Bcal_{j^\star,m_\ell}\bm v_\ell
	\longrightarrow
	\bm w_\infty
	\quad
	\text{strongly in }\Hcal.
	\label{eq:Bm_v_conv}
\end{equation}

Using the eigenvalue relation \eqref{eq:conv_eig_relation},
\begin{equation}
	\bm v_\ell
	=
	\rho_{j^\star,m_\ell}^{-1}
	\Bcal_{j^\star,m_\ell}\bm v_\ell.
	\label{eq:v_compact_repr}
\end{equation}
Together with \eqref{eq:rho_subseq_conv} and
\eqref{eq:Bm_v_conv}, this gives
\begin{equation}
	\bm v_\ell
	\longrightarrow
	\bm v_\infty
	\triangleq
	\rho_\infty^{-1}\bm w_\infty
	\quad
	\text{strongly in }\Hcal,
	\label{eq:v_strong_conv}
\end{equation}
Since $\norm{\bm v_\ell}_{\Hcal}=1$ for all $\ell$, continuity of the
norm implies
\begin{equation}
	\norm{\bm v_\infty}_{\Hcal}=1.
	\label{eq:v_inf_unit}
\end{equation}
Defining the limiting dual-slack operator $\Zcal_\infty \triangleq \mu_\infty\Ical - \Bcal_{j^\star,\infty}$, at iteration $m_\ell$, the violating direction satisfies
\begin{align}
	&
	\inner{\bm v_\ell}
	{
		\left(
		\mu_{m_\ell}\Ical
		-
		\Bcal_{j^\star,m_\ell}
		\right)
		\bm v_\ell
	}_{\Hcal}
	\nonumber\\
	&\qquad=
	\mu_{m_\ell}
	-
	\rho_{j^\star,m_\ell}
	=
	-\delta_{m_\ell}
	\le
	-\varepsilon.
	\label{eq:subseq_negative}
\end{align}
Taking $\ell\rightarrow\infty$ using
\eqref{eq:dual_subseq_conv},
\eqref{eq:B_op_convergence}, and
\eqref{eq:v_strong_conv} yields
\begin{equation}
	\inner{\bm v_\infty}
	{\Zcal_\infty\bm v_\infty}_{\Hcal}
	\le
	-\varepsilon.
	\label{eq:limit_negative}
\end{equation}

On the other hand, once $\bm v_\ell$ is identified at iteration
$m_\ell$, the enrichment rule ensures $	\bm v_\ell \in \Ucal_{m_\ell+1}$.
Since the search subspaces are nested, it follows that
\begin{equation}
	\bm v_\ell
	\in
	\Ucal_{m_r},
	\qquad
	r>\ell.
	\label{eq:old_v_in_later_space}
\end{equation}
Applying the restricted dual-slack feasibility condition
\eqref{eq:restricted_slack_on_Um} at iteration $m_r$ to the fixed vector $\bm v_\ell$ gives
\begin{equation}
	\inner{\bm v_\ell}
	{
		\left(
		\mu_{m_r}\Ical
		-
		\Bcal_{j^\star,m_r}
		\right)
		\bm v_\ell
	}_{\Hcal}
	\ge0.
	\label{eq:later_nonnegative}
\end{equation}
For each fixed $\ell$, letting $r\rightarrow\infty$ along the selected
subsequence and using~\eqref{eq:dual_subseq_conv} and
\eqref{eq:B_op_convergence}, we obtain
\begin{equation}
	\inner{\bm v_\ell}
	{\Zcal_\infty\bm v_\ell}_{\Hcal}
	\ge0.
	\label{eq:fixed_l_nonnegative}
\end{equation}
Then, letting $\ell\rightarrow\infty$ and using
\eqref{eq:v_strong_conv} gives
\begin{equation}
	\inner{\bm v_\infty}
	{\Zcal_\infty\bm v_\infty}_{\Hcal}
	\ge0,
	\label{eq:limit_nonnegative}
\end{equation}
which contradicts~\eqref{eq:limit_negative}. Therefore, $\delta_m\rightarrow0$.

Finally, by applying Proposition~\ref{prop:gap} and $\delta_m\rightarrow0$, we obtain \eqref{eq:pm_to_pinf}. This completes the proof.

\bibliographystyle{IEEEtran}
\bibliography{ref}

@article{zhao2026crbET,
	title={{Cramér-Rao Bound Optimization for Near-Field ISAC with Extended Targets}},
	author={Zhao, Zongyao and Wang, Zhaolin and Han, Lincong and Xu, Liang and Jin, Jing and Liu, Yuanwei and Huang, Kaibin},
	journal={arXiv preprint arXiv:2604.18166},
	year={2026}
}

@ARTICLE{wang2024tutorial,
	author={Wang, Zhe and Zhang, Jiayi and Du, Hongyang and Niyato, Dusit and Cui, Shuguang and Ai, Bo and Debbah, M{\'e}rouane and Letaief, Khaled B. and Poor, H. Vincent},
	journal={IEEE Commun. Surveys Tuts.}, 
	title={{A Tutorial on Extremely Large-Scale MIMO for 6G: Fundamentals, Signal Processing, and Applications}}, 
	year={2024},
	volume={26},
	number={3},
	pages={1560-1605},
	doi={10.1109/COMST.2023.3349276}}

@ARTICLE{guo2026deep,
	author={Guo, Jia and Liu, Yuanwei and Shin, Hyundong and Nallanathan, Arumugam},
	journal={IEEE Trans. Wireless Commun.}, 
	title={{Deep Learning for Beamforming in Multi-User Continuous Aperture Array Systems}}, 
	year={2026},
	volume={25},
	number={},
	pages={7773-7791},
	doi={10.1109/TWC.2025.3633569}}

@ARTICLE{zhang2026transformer,
	author={Zhang, Shumeng and Zhao, Yun and Zhang, Deyou and Chen, Maoxuan and Li, Qingchao and Liu, Yuanwei},
	journal={IEEE Wireless Commun. Lett.}, 
	title={{Transformer-Enabled Beamforming for Continuous Aperture Array Broadcast Communication Systems}}, 
	year={2026},
	volume={15},
	number={},
	pages={3402-3406},
	doi={10.1109/LWC.2026.3699435}}

@ARTICLE{ranasinghe2026doubly,
	author={Ranasinghe, Kuranage Roche Rayan and Wang, Zhaolin and Rou, Hyeon Seok and Abreu, Giuseppe Thadeu Freitas de and Björnson, Emil},
	journal={IEEE Trans. Wireless Commun.}, 
	title={{Doubly-Dispersive Continuous MIMO Systems: Channel Modeling and Beamforming Design}}, 
	year={2026},
	volume={25},
	number={},
	pages={15441-15458},
	doi={10.1109/TWC.2026.3682880}}

@ARTICLE{ouyang2025diversity,
	author={Ouyang, Chongjun and Wang, Zhaolin and Zhang, Xingqi and Liu, Yuanwei},
	journal={IEEE Trans. Wireless Commun.}, 
	title={{Diversity and Multiplexing for Continuous-Aperture Array (CAPA)-Based Communications}}, 
	year={2025},
	volume={24},
	number={10},
	pages={8662-8680},
	doi={10.1109/TWC.2025.3568247}}

@ARTICLE{zhao2025continuous,
	author={Zhao, Boqun and Ouyang, Chongjun and Zhang, Xingqi and Liu, Yuanwei},
	journal={IEEE Trans. Wireless Commun.}, 
	title={{Continuous-Aperture Array (CAPA)-Based Wireless Communications: Capacity Characterization}}, 
	year={2025},
	volume={24},
	number={12},
	pages={10456-10473},
	doi={10.1109/TWC.2025.3579677}}

@ARTICLE{rezvani2025energy,
	author={Rezvani, Maryam and Adve, Raviraj and Bin Sediq, Akram and El-Keyi, Amr},
	journal={IEEE J. Sel. Areas Inf. Theory}, 
	title={{Energy Efficient Wireless Communications by Harnessing Huygens’ Metasurfaces}}, 
	year={2025},
	volume={6},
	number={},
	pages={85-99},
	doi={10.1109/JSAIT.2025.3566321}}

@ARTICLE{liu2025capa,
	author={Liu, Yuanwei and Ouyang, Chongjun and Wang, Zhaolin and Xu, Jiaqi and Mu, Xidong and Ding, Zhiguo},
	journal={IEEE Wireless Commun.}, 
	title={{CAPA}: Continuous-Aperture Arrays for Revolutionizing {6G} Wireless Communications}, 
	year={2025},
	volume={32},
	number={4},
	pages={38-45},
	doi={10.1109/MWC.001.2400493}}

@book{kreyszig1978introductory,
	title     = {Introductory Functional Analysis with Applications},
	author    = {Kreyszig, Erwin},
	year      = {1978},
	publisher = {John Wiley \& Sons},
	address   = {New York}
}

@ARTICLE{pizzo2022fourier,
	author={Pizzo, Andrea and Sanguinetti, Luca and Marzetta, Thomas L.},
	journal={IEEE Trans. Wireless Commun.}, 
	title={{Fourier Plane-Wave Series Expansion for Holographic MIMO Communications}}, 
	year={2022},
	volume={21},
	number={9},
	pages={6890-6905},
	doi={10.1109/TWC.2022.3152965}}

@ARTICLE{sanguinetti2023wavenumber,
	author={Sanguinetti, Luca and D’Amico, Antonio Alberto and Debbah, Merouane},
	journal={IEEE Trans. Wireless Commun.}, 
	title={{Wavenumber-Division Multiplexing in Line-of-Sight Holographic MIMO Communications}}, 
	year={2023},
	volume={22},
	number={4},
	pages={2186-2201},
	doi={10.1109/TWC.2022.3208961}}

@ARTICLE{zhang2023pattern,
	author={Zhang, Zijian and Dai, Linglong},
	journal={IEEE J. Sel. Areas Commun.}, 
	title={{Pattern-Division Multiplexing for Multi-User Continuous-Aperture MIMO}}, 
	year={2023},
	volume={41},
	number={8},
	pages={2350-2366},
	doi={10.1109/JSAC.2023.3288244}}

@ARTICLE{wang2025beamforming,
	author={Wang, Zhaolin and Ouyang, Chongjun and Liu, Yuanwei},
	journal={IEEE Trans. Wireless Commun.}, 
	title={{Beamforming Optimization for Continuous Aperture Array (CAPA)-Based Communications}}, 
	year={2025},
	volume={24},
	number={6},
	pages={5099-5113},
	doi={10.1109/TWC.2025.3545770}}

@ARTICLE{wang2025optimal,
	author={Wang, Zhaolin and Ouyang, Chongjun and Liu, Yuanwei},
	journal={IEEE Trans. Commun.}, 
	title={Optimal Beamforming for Multi-User Continuous Aperture Array (CAPA) Systems}, 
	year={2025},
	volume={73},
	number={10},
	pages={9207-9221},
	doi={10.1109/TCOMM.2025.3554644}}

@INPROCEEDINGS{wang2025continuous,
	author={Wang, Zhaolin and Ouyang, Chongjun and Zou, Yixuan and Liu, Yuanwei},
	booktitle={ICC 2025 - IEEE International Conference on Communications}, 
	title={{Continuous Aperture Array (CAPA) Beamforming: A Calculus of Variations Method}}, 
	year={2025},
	volume={},
	number={},
	pages={6267-6271},
	doi={10.1109/ICC52391.2025.11161275}}

@ARTICLE{zhang2025exploiting,
	author={Zhang, Yue and Ouyang, Chongjun and Shan, Hangguan and Liu, Yuanwei and Zhou, Yong and Shi, Zhiguo},
	journal={IEEE Trans. Wireless Commun.}, 
	title={{Exploiting Continuous-Aperture Arrays in Integrated Sensing and Communication Systems}}, 
	year={2025},
	volume={24},
	number={11},
	pages={9539-9555},
	doi={10.1109/TWC.2025.3573872}}

@ARTICLE{zhang2025RateCRB,
	author={Zhang, Yue and Shan, Hangguan and Ouyang, Chongjun and Liu, Yuanwei and Shi, Zhiguo and Lin, Dong and Hou, Fen},
	journal={IEEE Trans. Wireless Commun.}, 
	title={{Continuous-Aperture Array for Integrated Sensing and Communication: Rate-CRB Tradeoff}}, 
	year={2026},
	volume={25},
	number={},
	pages={4215-4230},
	doi={10.1109/TWC.2025.3609401}}

@ARTICLE{zhao2026downlink,
	author={Zhao, Boqun and Ouyang, Chongjun and Zhang, Xingqi and Shin, Hyundong and Liu, Yuanwei},
	journal={IEEE Trans. Wireless Commun.}, 
	title={{Downlink and Uplink ISAC in Continuous-Aperture Array (CAPA) Systems}}, 
	year={2026},
	volume={25},
	number={},
	pages={3592-3609},
	doi={10.1109/TWC.2025.3605158}}

@ARTICLE{ye2026optimal,
	author={Ye, Junjie and Wang, Zhaolin and Liu, Yuanwei and Zhang, Peichang and Huang, Lei and Nallanathan, Arumugam},
	journal={IEEE Trans. Wireless Commun.}, 
	title={{Optimal Waveform Design for Continuous Aperture Array (CAPA)-Aided ISAC Systems}}, 
	year={2026},
	volume={25},
	number={},
	pages={19082-19098},
	doi={10.1109/TWC.2026.3702839}}

@ARTICLE{du2023integrated,
	author={Du, Zhen and Liu, Fan and Yuan, Weijie and Masouros, Christos and Zhang, Zenghui and Xia, Shuqiang and Caire, Giuseppe},
	journal={IEEE Trans. Wireless Commun.}, 
	title={{Integrated Sensing and Communications for V2I Networks: Dynamic Predictive Beamforming for Extended Vehicle Targets}}, 
	year={2023},
	volume={22},
	number={6},
	pages={3612-3627},
	doi={10.1109/TWC.2022.3219890}}

@ARTICLE{wang2024cramer,
	author={Wang, Yiqiu and Tao, Meixia and Sun, Shu},
	journal={IEEE Trans. Wireless Commun.}, 
	title={{Cram{\'e}r-Rao Bound Analysis and Beamforming Design for Integrated Sensing and Communication With Extended Targets}}, 
	year={2024},
	volume={23},
	number={11},
	pages={15987-16000},
	doi={10.1109/TWC.2024.3435864}}

@ARTICLE{wang20263d,
	author={Wang, Yiqiu and Tao, Meixia and Sun, Shu and Mo, Jianhua and Cao, Wei},
	journal={IEEE Trans. Wireless Commun.}, 
	title={{3D Extended Target Sensing in ISAC: Cramér-Rao Bound Analysis and Beamforming Design}}, 
	year={2026},
	volume={25},
	number={},
	pages={20242-20258},
	doi={10.1109/TWC.2026.3709579}}

@ARTICLE{zhao2026continuous,
	author={Zhao, Boqun and Ouyang, Chongjun and Zhang, Xingqi and Liu, Yuanwei},
	journal={IEEE Trans. Wireless Commun.}, 
	title={{Continuous-Aperture Array-Based ISAC Over Fading Channels}}, 
	year={2026},
	volume={25},
	number={},
	pages={22494-22508},
	doi={10.1109/TWC.2026.3726026}}

@ARTICLE{liu2026symbol,
	author={Liu, Hongli and Li, Qiang},
	journal={IEEE Wireless Commun. Lett.}, 
	title={{Symbol-Level Precoding for Continuous-Aperture ISAC Systems}}, 
	year={2026},
	volume={15},
	number={},
	pages={5616-5620},
	doi={10.1109/LWC.2026.3736092}}

@ARTICLE{Garcia2022crb,
	author={Garcia, Nil and Fascista, Alessio and Coluccia, Angelo and Wymeersch, Henk and Aydogdu, Canan and Mendrzik, Rico and Seco-Granados, Gonzalo},
	journal={IEEE Trans. Signal Process.}, 
	title={Cramér-Rao Bound Analysis of Radars for Extended Vehicular Targets With Known and Unknown Shape}, 
	year={2022},
	volume={70},
	number={},
	pages={3280-3295},
	doi={10.1109/TSP.2022.3183853}}
\end{document}